\documentclass[preprint,12pt]{elsarticle}

\usepackage{amssymb}
\usepackage{amsmath}
\usepackage{amsthm}

\journal{Information and Computation}
\usepackage[capitalise]{cleveref}
\usepackage{thm-restate}
\usepackage[title]{appendix}
\usepackage{mdframed,framed}
\usepackage{bm}
\usepackage{xspace}
\usepackage{microtype}
\usepackage{xcolor}
\usepackage{tikz}
\usetikzlibrary{calc,automata,positioning,decorations.pathreplacing}
\tikzset{align at top/.style={baseline=(current bounding box.north)}}
\tikzstyle{every node}=[font=\scriptsize]
\tikzstyle{state} = [draw,fill=white,circle,thick,align=center,inner sep=0pt,minimum size=4.5mm]
\tikzstyle{lstate} = [draw,fill=white,rectangle,rounded corners,thick,align=center,inner sep=2pt]
\tikzstyle{dot} = [fill,circle,inner sep=0mm,minimum size=1.25mm,line width=0mm]
\usepackage[colorinlistoftodos]{todonotes}
\usepackage{etoolbox}
\usepackage{mathtools}
\usepackage{aliascnt}
\usepackage{crossreftools}
\usepackage{xifthen}
\usepackage{diagbox}
\usepackage{threeparttable}
\usepackage{multirow}
\usepackage{url}
\usepackage{wrapfig}
\usepackage{todonotes}

\usepackage{enumitem}

\usepackage{float}
\usepackage[noend]{algorithmic} 
\usepackage[ruled]{algorithm}
\usepackage[capitalise]{cleveref}
\usepackage{pgfplots}
\usetikzlibrary{arrows.meta}

\theoremstyle{plain}
	\newtheorem{theorem}{Theorem}[section]
	\newtheorem{lemma}[theorem]{Lemma}
	\newtheorem{corollary}[theorem]{Corollary}
	\newtheorem{proposition}[theorem]{Proposition}
	
	\newtheorem{definition}[theorem]{Definition}

\theoremstyle{definition}
	\newtheorem{example}[theorem]{Example}
\theoremstyle{remark}
	
	\newtheorem*{note*}{Note}
	\newtheorem{remark}[theorem]{Remark}
	\newtheorem*{remark*}{Remark}
\theoremstyle{claimstyle}
	
	\newtheorem*{claim*}{Claim}

\newcommand{\mdp}{\mathcal{M}}
\newcommand{\prob}{\mathbb{P}}
\newcommand{\RR}{\mathbb{R}}
\newcommand{\QQ}{\mathbb{Q}}
\newcommand{\NN}{\mathbb{N}}
\newcommand{\pathset}[2]{\mathrm{Path}(#1, #2)}

\newcommand{\eps}{\bm{\varepsilon}}

\DeclareMathOperator*{\argmin}{arg\,min}
\DeclareMathOperator*{\argmax}{arg\,max}

\DeclareMathOperator*{\supp}{supp}
\DeclareMathOperator*{\suc}{succ}

\newcommand{\anote}[2][]{
	\reversemarginpar } %
\newcommand{\knote}[2][]{
	\reversemarginpar } %

\definecolor{gryffindor}{RGB}{220,0,1}
\definecolor{ravenclaw}{RGB}{14,26,164}
\definecolor{slytherin}{RGB}{42,98,61}

\begin{document}

\begin{frontmatter}

\title{On Piecewise Affine Reachability\\ with Bellman Operators}

\affiliation[label:tu]{organization={TU Wien},%
	country={Austria}}
\affiliation[label:nii]{organization={National Institute of Informatics},%
	country={Japan}}
\affiliation[label:sokendai]{organization={The Graduate University for Advanced Studies (SOKENDAI)},%
	country={Japan}}

\author[label:tu]{Anton Varonka\corref{cor1}} %
\author[label:nii,label:sokendai]{Kazuki Watanabe} %
\cortext[cor1]{Corresponding author. E-mail: anton.varonka@gmail.com}
\begin{abstract}
 We study the following reachability problem for piecewise affine maps: Given two vectors $\bm{s}, \bm{t} \in \QQ^d$ and a piecewise affine map $f\colon \QQ^d\rightarrow \QQ^d$, does there exist $n\in \NN$ such that $f^{n}(\bm{s}) = \bm{t}$?
In this work, we focus on  this reachability problem for a subclass of piecewise affine maps---Bellman operators arising from Markov decision processes. 
 We prove that the reachability problem for $\max$- and $\min$-Bellman operators is decidable in any dimension under either of the following conditions: (i) the target vector~$\bm{t}$ is not the fixed point of the operator~$f$; or (ii) the initial and target vectors~$\bm{s}$ and~$\bm{t}$ are comparable with respect to the componentwise order.
Furthermore, we show that in the two-dimensional case, the reachability problem for Bellman operators is decidable for arbitrary~$\bm{s}, \bm{t} \in \QQ^2$.
This stands in sharp contrast to the known undecidability of reachability for general piecewise affine maps in dimension $d = 2$. 
\end{abstract}


\end{frontmatter}

\section{Introduction}
Solving reachability problems is central to formal verification, but it is challenging, as evidenced by a body of work that has been pursued for decades.
Specifically, the reachability problem that we are interested in asks the following question: 
Given two vectors $\bm{s}, \bm{t}$ and a map $f$, is there $n\in  \NN$ such that $f^n(\bm{s}) = \bm{t}$? 
One of the seminal results is given by Kannan and Lipton~\cite{kannan1980decid,kannan1986poly}.
They answer the decidability question in the affirmative---by presenting a polynomial-time algorithm---for \emph{affine maps} $f$ and vectors $\bm{s}, \bm{t}$ with rational coefficients.\footnote{The results by Kannan and Lipton hold for linear maps. Since the reachability problem is decidable in arbitrary dimension, a standard technique extends it to the affine maps by encoding a $d$-dimensional affine map as a linear map in dimension~$d+1$; see also~\cite[Section~5]{Tiwari2005termination}.}

Unfortunately, the reachability problem becomes undecidable by slightly extending the class of maps beyond the affine maps---\emph{piecewise affine maps (PAMs)}. 
Koiran et al.~\cite{koiran1994pam} show that the reachability problem for PAMs is undecidable even in the two-dimensional space, 
which witnesses the significant difficulty of the problem, compared to that of affine maps. 
Specifically, a PAM~$f$ on the domain~$\mathcal{D}$ is a function with the property that
for some family~$\{P_1, \dots, P_k\}$ of sets 
such that $P_1 \cup \dots \cup P_k = \mathcal{D}$, 
the restriction of~$f$ to each~$P_i$ is an affine function.
We consider piecewise affine maps over the domain~$\mathcal{D} = [0,1]^d$, where $[0,1]$ is the unit interval.
Each of the finitely many \emph{pieces}~$P_1, \dots, P_k$ is defined by a conjunction of finitely many linear inequalities.

\begin{example}\label{ex:intro}
	Consider an example of a PAM in the dimension~$d=2$. Let
\(f : [0,1]\times[0,1] \rightarrow  [0,1]\times[0,1]\) be defined by $f(x_1, x_2) = (x_1', x_2')$ with
\begin{equation*}
	x_1' = \frac{1}{2}x_2 + \frac{1}{3}, \hspace{4cm}
	x_2' =	\begin{cases}
		\frac{1}{2}x_1 + \frac{1}{2},\quad &\text{ if $x_1 \geq x_2$},\\
		\frac{1}{4}x_1 + \frac{1}{4}x_2 + \frac{1}{2}, \quad &\text{ if $x_1 < x_2$.}
	\end{cases}
\end{equation*}
Here, $P_1$ is defined as $\{\bm{x} = (x_1, x_2) : x_1 \geq x_2\}$ and $P_2$ as $\{\bm{x} = (x_1, x_2) : x_1 < x_2\}$,
where according to the definition $P_1 \cup P_2 = \mathcal{D}$.
\end{example}

In fact, piecewise affine maps characterise a very rich mathematical object~\cite{Bell2008undec,blondel2000survey}.
The reachability problem for PAMs is thus one of a series of challenging problems
whose crux lies in the unpredictable behaviour of iterative maps 
and corresponding discrete-time dynamical systems, see also~\cite{Bell2008undec,ChonevOW15,Karimov2022dynsys,Tiwari2005termination}.

A natural question might be: \emph{When does the reachability problem for PAMs become decidable?}	
Indeed, the reachability problem for PAMs on a unit interval (dimension one) has received attention in recent research~\cite{bournez2018onedpam,ghahremani2023injective,kuijpers2021skew},
where classes of PAMs with decidable reachability problems have been found
while fostering new elaborate techniques.
Yet, the decidability of the general one-dimensional reachability problem for PAMs remains open,
even for maps defined with two pieces.

\subparagraph{Our Approach.}
In this work, we propose an orthogonal approach to investigate the challenges behind the reachability problem for PAMs, focusing not on the restriction of the \emph{dimension} or the number of \emph{pieces},
but on the restriction of the \emph{structure} of PAMs.
Specifically, we consider the reachability problem for the \emph{Bellman operators} on \emph{Markov decision processes (MDPs)}---Bellman operators are, in fact, PAMs that have been studied mostly in the context of software verification or reinforcement learning~\cite{bk,bellman1957mdp,Puterman94}.
For instance, the PAM in~\cref{ex:intro} is a Bellman operator.

MDPs are a standard probabilistic model for systems with uncertainties, and the least fixed points of the Bellman operators $\Phi\colon [0, 1]^{d}\rightarrow [0, 1]^{d}$ represent the ``optimal'' reachability probability to the specific target state. 
Here, optimal means that the maximum/minimum reachability probability is induced by a \emph{scheduler} that resolves the non-deterministic behaviour on MDPs. 
We formulate our target problem as follows: Given two vectors $\bm{s}, \bm{t}\in [0, 1]^{d}$ and a Bellman operator $\Phi\colon [0, 1]^{d}\rightarrow [0, 1]^{d}$, is there $n\in \NN$ such that $\Phi^{n}(\bm{s}) = \bm{t}$?
We refer to this problem as the \emph{reachability problem for Bellman operators} (or BOR, for Bellman Operator Reachability).

Under a reasonable assumption, from any vector $\bm{s}$, the sequence $\langle\Phi^{n}(\bm{s})\rangle_{n\in \NN}$ converges to the unique fixed point $\mu\Phi$.
Here, the unique fixed point is precisely the vector of optimal reachability probabilities from each state to the target state.
The iterative procedure where~$\mu\Phi$ is approximated by applying~$\Phi$ is referred to as~\emph{value iteration} and is widely studied~\cite{bk,Baier0L0W17,chatterjee2008vi}.
Note that the unique fixed point is computable in polynomial time by linear programming~\cite{bk}, so for our problem we can assume that we know the unique fixed point $\mu\Phi$ a priori. 
Nevertheless, even a convergent sequence does not generally reach~$\mu\Phi$ at any~$n$.
In fact, the question of the reachability to the fixed point (in finite time) is known in both theoretical computer science and software verification communities.
A case in point is the discussion in~\cite{koiran1994pam}, where this question is explicitly asked for one-dimensional PAMs (see also~\cite{blondel2001note}).
Furthermore, in a value iteration survey~\cite{chatterjee2008vi} the same property is listed.
While the authors observe that fixed points are not reachable in general, they do not discuss the decidability aspect.
In the present paper, we investigate the decidability of reaching~$\bm{t} = \mu\Phi$ as part of our problem.

\subparagraph{Contributions.}
We present some decidability results for our target problem under a condition that ensures the existence of the unique fixed point $\mu\Phi$ of the Bellman operators~\cite{HaddadM18}. 
First, we show that the reachability problem for Bellman operators is decidable for any dimension if the target vector $\bm{t}$ does not coincide with $\mu\Phi$ ($\bm{t}\not = \mu\Phi$). 
The proof uses the fact that from any vector $\bm{s}$, the iteration of the Bellman operator converges to $\mu\Phi$.
It becomes rather non-trivial when $\bm{t} = \mu\Phi$, that is, for the reachability problem to the unique fixed point $\mu\Phi$. 
We show that the reachability problem for Bellman operators when $\bm{t} = \mu\Phi$ is decidable for any dimension if $\bm{s}$ is comparable to $\mu\Phi$, that is, either $\bm{s}\leq \mu\Phi$ or $\mu\Phi \leq \bm{s}$ holds for the componentwise order.  
The crux is to show that eventually only ``optimal'' actions are chosen, and we reduce the reachability problem to a simple qualitative reachability problem that can be shown decidable. 

Finally, we address the remaining case: $\bm{t} = \mu\Phi$ and $\bm{s}$ is incomparable to $\mu\Phi$. 
In dimension~$2$, we show an algorithmic procedure also for this case---finding the last piece of the puzzle---the reachability problem for two-dimensional Bellman operators is thus decidable. 
Our argument is based on analysing the equivalent problem for matrix semigroups, and 
our proof exploits 
the existence of a total order on the lines induced by actions.
To the best of our knowledge, this is the first result to give a reasonably large class of PAMs for which the reachability problem is decidable in the two-dimensional case. 

This article is the extended journal version of~\cite{VaronkaW25}, which studied the piecewise affine reachability problem for the Bellman operator with the \emph{maximum} objective.
A novel contribution of the current work is investigating the corresponding problem for the Bellman operator with the \emph{minimum} objective.
To this end, we introduce additional preliminaries for the minimum objective in~\cref{sec:prelim} and extend the proofs in~\cref{sec:efs,sec:2d} so that they uniformly handle both Bellman operators.
We refined the presentation of the results, using $\Phi$ which refers to either $\Phi_{\max}$ or $\Phi_{\min}$, 
where appropriate. 
The main theorem of~\cref{sec:efs}, \cref{compar-decid}, is stated in terms of both Bellman operators.
In~\cref{sec:2d}, we introduce new definitions to substantially simplify the case distinction of the proofs in~\cite{VaronkaW25}. 
The notion of kernel lines helps identify the core idea of the new proof. 
Thus, the decidability of the BOR problem for dimension~2 is proved for both maximum and minimum objectives.

\subparagraph{Organization.}
We outline our paper as follows. 
\begin{itemize}
	\item In~\cref{sec:prelim}, we formally define the main objects of our study, MDPs and Bellman operators, and recall some known properties. 
	\item In~\cref{sec:efs}, we show that the reachability problem for Bellman operators is decidable when either $\bm{t} \not = \mu\Phi$ (\cref{t-not-fp}) or $\bm{s}$ is comparable to $\bm{t} = \mu\Phi$ (\cref{compar-decid}). 
	Importantly, the result holds for arbitrary dimension. 
	\item In~\cref{sec:2d}, we prove that the reachability problem for Bellman operators is decidable in the two-dimensional case (\cref{2d-efs}), by presenting an algorithm that solves the remaining case ($\bm{s}$ is incomparable to $\bm{t} = \mu\Phi$). 
	\item In~\cref{relatedwork}, we discuss our result with related work, and list some future directions. 
\end{itemize}

\section{Background}\label{sec:prelim}

We first recall some definitions and properties for Markov decision processes (MDPs) and their Bellman operators, which are necessary for our development. 
We then define our target problem, namely the piecewise affine reachability problem with Bellman operators. 

\subsection{Preliminary}
\begin{definition}[MDP~\cite{Puterman94}]
	An MDP $\mdp$ is a tuple $(S, Act, \prob)$ such that (i) $S$ is a finite non-empty set of \emph{states}; 
	(ii) $Act$ is an indexed family $(Act_s)_{s\in S}$ of finite sets of \emph{actions} such that the sets $Act_s$ and $Act_{s'}$ of actions on $s$ and $s'$ are disjoint for any $s, s'\in S$; and (iii) $\prob$ is the transition probability function~$\prob(s, \alpha, \_) \in S\rightarrow [0, 1]\cap \QQ$ with finite support that satisfies $\sum_{s'\in S} \prob(s, \alpha, s') = 1$, for any $s\in S$ and $\alpha \in Act_s$. 
\end{definition}

We refer to the support of $\prob(s, \alpha, \_)$ by $\supp(s, \alpha)$.
We fix a \emph{target state}~$t$ and assume that~$t$ is a sink, i.e.,~$Act_t = \varnothing$.
We also assume that there is a distinguished sink state $s_{-} (\not = t)$.
\begin{example}
	\label{ex:mdp}
	We present an MDP $\mdp = (S, Act, \prob)$, where (i) $S = \{s_1, s_2, s_3,t\}$; (ii) $Act_{s_1} \coloneqq \{\alpha\}$, $Act_{s_2} \coloneqq \{\beta_1, \beta_2\}$, and $Act_{s_3} = Act_{t} \coloneqq \varnothing$; and (iii) $\prob$ is defined by 
	\begin{align*}
		&\prob(s_1, \alpha, s_2) \coloneqq 1/2,\quad \prob(s_1, \alpha, s_3) \coloneqq 1/6,\quad \prob(s_1, \alpha, t) \coloneqq 1/3,\\
		&\prob(s_2, \beta_1, s_1)\coloneqq 1/2, \quad \prob(s_2, \beta_1, t) \coloneqq 1/2, \\
		&\prob(s_2, \beta_2, s_1)\coloneqq 1/4, \quad \prob(s_2, \beta_2, s_2) \coloneqq 1/4, \quad \prob(s_2, \beta_2, t) \coloneqq 1/2. 
	\end{align*}
\end{example}

Given states~$s,s' \in S$,
a \emph{path} $\pi$ from $s$ to $s'$ is a sequence $\pi \coloneqq (s_1,\dots, s_m)$ such that $s_i\in S\setminus \{t\} $ for any $i\in [1, m-1]$, $s_1 = s$, and $s_m  = s'$.
We denote the set of paths from~$s$ to~$s'$ by $\pathset{s}{s'}$. 
A \emph{scheduler} is a function $\sigma\colon S^{+}\rightarrow \cup_{s\in S}Act_s$ such that $\sigma(s_1\cdots s_m)\in Act_{s_m}$. %
As deterministic schedulers suffice for the reachability objective~\cite{bk}, we further consider
the set~$\Sigma$ of all deterministic schedulers. 
A scheduler is \emph{positional} if for any $s_1\cdots s_m\cdot s$ and $s'_1\cdots s'_n\cdot s$, the actions $\sigma(s_1\cdots s_m\cdot s)$ and 
$\sigma(s'_1\cdots s'_n\cdot s)$ coincide. 
For a path $\pi \coloneqq (s_1,\dots, s_m)$ and a scheduler $\sigma \in \Sigma$, %
define $\prob^{\sigma}(\pi)\coloneqq \prod_{i\in [1, m-1]} \prob(s_i, \sigma(\pi_i), s_{i+1})$, where %
$\pi_i = (s_1,\dots, s_i)$. 
\begin{definition}[reachability probability]
	Given a scheduler $\sigma$, and a state $s\in S$, the \emph{reachability probability} $\prob^\sigma(s \models \lozenge t)$ under $\sigma$ is defined by 
	\[\prob^\sigma(s \models \lozenge t)\coloneqq \sum_{\pi \in \pathset{s}{t}} \prob^{\sigma}(\pi).\] 
	
	The \emph{maximum reachability probability} is defined as \[p^{\mathrm{max}}_s:=\sup_{\sigma\in \Sigma} \prob^\sigma(s \models \lozenge t) \in \QQ,\] 
	and the \emph{minimum reachability probability} is defined as \[p^{\mathrm{min}}_s:=\inf_{\sigma\in \Sigma} \prob^\sigma(s \models \lozenge t) \in \QQ.\] 
\end{definition}
We write~$\bm{p^{\max}}$ and $\bm{p^{\min}}$ for the vectors $(p^{\max}_s)_{s \in S\backslash \{t\}}$ and $(p^{\min}_s)_{s \in S\backslash \{t\}}$ indexed by states~$S\backslash \{t\}$, respectively.
The maximum and minimum reachability probabilities are in fact achievable by a positional scheduler. 
\begin{proposition}[e.g.~\cite{bk}] %
	\label{vi-positional}
	There exist optimal positional schedulers $\sigma^{\max}_\text{pos}, \sigma^{\min}_\text{pos} \in \Sigma$ such that
	\[\prob^{\sigma^{\max}_\text{pos}}(s \models \lozenge t) = p^{\max}_s \text{ and } \;
	\prob^{\sigma^{\min}_\text{pos}}(s \models \lozenge t) = p^{\min}_s\] hold for all~$s \in S$, respectively.
\end{proposition}
\emph{Value Iteration (VI)}~\cite{bk,Puterman94} is a standard technique to approximate the vector of optimal (maximum or minimum) reachability probabilities. 
Specifically, VI applies the \emph{Bellman operator} to the current approximation for each iteration step. 

Let $S_d = S \setminus \{t, s_{-}\}$. 
For each $s\in S_d$, we associate a polynomial of degree~1, \emph{a linear polynomial} of~$\alpha$, with each action $\alpha \in Act_s$. This polynomial $L_\alpha \in \QQ[\bm{x}]$ is defined as 
\[L_\alpha(\bm{x}) = \sum_{s' \in S_d} \prob(s, \alpha, s') x_{s'}  + \prob(s, \alpha, t).\]%
\begin{definition}[Bellman operator]
	The \emph{Bellman operator} $\Phi_{\max} \colon [0, 1]^d\rightarrow [0, 1]^d$ \emph{for the maximum reachability objective} is defined by \[\Phi_{\max}(\bm{x})_s \coloneqq \max_{\alpha\in Act_s} L_\alpha(\bm{x})\]
	for each $\bm{x} = (x_s)_{s \in S_d}\in [0, 1]^d$ and $s\in S_d$.

	The \emph{Bellman operator} $\Phi_{\min} \colon [0, 1]^d\rightarrow [0, 1]^d$ \emph{for the minimum reachability objective} is defined by \[\Phi_{\min}(\bm{x})_s \coloneqq \min_{\alpha\in Act_s} L_\alpha(\bm{x})\]
	for each $\bm{x} = (x_s)_{s \in S_d}\in [0, 1]^d$ and $s\in S_d$.
\end{definition}
\begin{example}
	The Bellman operator $\Phi_{\max}$ of the MDP given in~\cref{ex:mdp} is given by 
	\begin{align*}
		\Phi_{\max}(\bm{x})_{s_1} &\coloneqq 1/2\cdot x_2 + 1/3, \\
		\Phi_{\max}(\bm{x})_{s_2} &\coloneqq \max\big(1/2 \cdot x_1 + 1/2,\, 1/4\cdot x_1 + 1/4\cdot x_2 + 1/2 \big),
	\end{align*}
	where $S_d =\{s_1, s_2\}$. The Bellman operator $\Phi_{\max}$ is indeed the PAM $f$ given in~\cref{ex:intro}. 
	
	The Bellman operator~$\Phi_{\min}$ of the same MDP is given by
	\begin{align*}
		\Phi_{\min}(\bm{x})_{s_1} &\coloneqq 1/2\cdot x_2 + 1/3, \\
		\Phi_{\min}(\bm{x})_{s_2} &\coloneqq \min\big(1/2 \cdot x_1 + 1/2,\, 1/4\cdot x_1 + 1/4\cdot x_2 + 1/2 \big),
	\end{align*}
	where $S_d =\{s_1, s_2\}$.
\end{example}

By an abuse of notation, we write $Act_i$ for the set $Act_{s_i}$ of actions on $s_i\in S_d$, and the optimal reachability probability $p_i$ for $p_{s_i}$. 
We also restrict~$\bm{p^{\max}}$ to the vector over $S_d$ and write $\bm{p^{\max}} = (p_1, \dots, p_d)$ and same for $\bm{p^{\min}}$ when it is clear from the context.
Each action $\alpha \in Act_i$ is associated with a set $\suc(\alpha)$ of successor states defined as $\suc(\alpha):= \supp(s_i, \alpha) \cap S_d$. We emphasise that the successor set, together with transition probabilities~$\prob(s, \alpha, s')$, $s' \in \suc(\alpha)$, is a probabilistic \emph{sub}distribution.

We define the partial order $\leq$ on vectors in $\RR^d$ by $\bm{u}\leq \bm{v}$ if $u_i\leq v_i$ holds for each $i\in [1, d]$. 
We refer to vectors~$\bm{u},\bm{v} \in \RR^d$ as \emph{comparable} if either $\bm{u} \geq \bm{v}$ or $\bm{u} \leq \bm{v}$ holds. 
Otherwise, the vectors are incomparable, denoted~$\bm{u} \bowtie \bm{v}$.

Let $||\cdot||_\infty$ denote the \emph{$\ell^\infty$-norm}, or the \emph{$\max$-norm}, defined by \(||\bm{x}||_\infty := \max{\left(|x_1|, \dots, |x_d|\right)}\) for a vector~$\bm{x} = (x_1, \dots, x_d)$.
In the sequel, the notation $||\bm{x}||$ stands for $||\bm{x}||_\infty$.
We further define the $\ell^\infty$-metric, i.e., the distance between two vectors $\bm{x}$ and $\bm{y}$ is
\(d(\bm{x},\bm{y}):= \max_i{\left(|x_i - y_i|\right)}.\)

The set $[0, 1]^d$ is a complete lattice with the componentwise ordering and the Bellman operator is $\omega$-continuous, that is, it preserves joins of ascending sequences.
Due to this, the iterative updates by the Bellman operators~$\Phi_{\max}$ and $\Phi_{\min}$ from the bottom vector~$\bm{0} \coloneqq (0, \cdots, 0)\in [0, 1]^d$, which are called VI (from $\bm{0}$), converge to the least fixed points~$\mu\Phi_{\max}$ and $\mu\Phi_{\min}$, which are the optimal reachability probabilities $\bm{p^{\max}}$ and $\bm{p^{\min}}$, respectively. 
\begin{proposition}[\cite{bk,chatterjee2008vi,Puterman94}]
	The sequences \[\langle(\Phi_{\max})^n(\mathbf{0})\rangle_{n\in \NN} \; \text{ and } \; \langle(\Phi_{\min})^n(\mathbf{0})\rangle_{n\in \NN}\] are monotonically increasing and converge to the optimal reachability probabilities $\bm{p^{\max}}$ and $\bm{p^{\min}}$, respectively.  
\end{proposition}

Furthermore, under the assumption of no existence of \emph{end components}, the Bellman operators have unique fixed points, and from any initial vectors the iteration converges to unique fixed points~\cite{HaddadM18}. 
\begin{definition}[end component~\cite{Alfaro97}]
	Let $\mdp = (S, Act, \prob)$ be an MDP. 
	A pair $(S', Act')$ such that $\varnothing \not =S' \subseteq S$ and $ \varnothing \not = Act'\subseteq \cup_{s\in S'}Act_s $ is an \emph{end component}
	if (i) for all $s\in S'$ and $\alpha \in Act'\cap Act_s$, $\supp(s, \alpha) \subseteq S'$; and (ii) the directed graph that is induced by $(S', Act')$ is strongly connected.  
\end{definition}
\begin{proposition}[\cite{HaddadM18}]\label{all-converge}
	Assume $\Phi$ is a Bellman operator of an MDP with no end components. 
	The sequence $\langle\Phi^n(\bm{x})\rangle_{n\in \NN}$ converges to the unique fixed point $\mu \Phi$ from any initial vector $\bm{x}\in [0, 1]^d$. 
\end{proposition}
In the sequel, we always assume that an MDP~$\mdp$ has no end components and hence its Bellman operators~$\Phi_{\max}$ and $\Phi_{\min}$ have unique fixed points.
Such MDPs remain expressive and exhibit highly non-trivial behaviors, making them an important subject of study in probabilistic verification; see e.g.~\cite{HaddadM18,BrazdilCCFKKPU14,QuatmannK18,hartmanns2020optimistic,Baier0L0W17}.
Notably, the presence of end components in a given MDP can be checked in P~\cite{Alfaro97}.
Moreover, \cite{HaddadM18}
describes reductions for $\Phi_{\max}$ and $\Phi_{\min}$ that, for an arbitrary MDP, constructs MDPs with the same least fixed points and without end components.
For $\Phi_{min}$, all maximal end components are simply merged into the distinguished sink state $s_{-}$. 
For $\Phi_{max}$, bottom maximal end components are merged into the distinguished sink state 
$s_{-}$, while each of the remaining maximal end components is shrunk into a single state, respectively.
These reductions never increase the dimension.
\subsection{Target problem}
We begin by recalling the value iteration algorithm.
As an iterative procedure, VI boils down to repeatedly applying the Bellman operator~$\Phi$ starting from a certain vector~$\bm{s} \in [0,1]^d$ (usually, $\bm{s} = \bm{0}$) 
and converging to the least fixed point~$\mu\Phi$---this becomes the unique fixed point in our setting by assuming the non-existence of end components.

Observe that every Bellman operator~$\Phi$
is indeed a piecewise affine map on the domain~$[0,1]^d$.
On close inspection, for every~$\bm{x} \in [0,1]^d$, the value of $\Phi^{\max}(\bm{x})$ is computed as the %
maximum of finitely many affine functions~$\phi_1, \dots, \phi_k : [0,1]^d \rightarrow [0,1]^d$ evaluated at~$\bm{x}$.
Let~$P_i \subseteq [0,1]^d$ be the set of points~$\bm{x} \in [0,1]^d$, for which %
$\phi_i(\bm{x}) \geq \phi_j(\bm{x})$ for each $j \neq i$.
Two observations are straightforward: 1) each~$P_i$ is defined by a conjunction of linear inequalities; 2) every~$\bm{x}$ belongs to at least one set~$P_i$.
Moreover, $\Phi_{\max}$ is a well-defined function,
which can be observed from $\bm{x} \in P_i \cap P_j$ implying $\phi_i(\bm{x}) = \phi_j(\bm{x})$.
Therefore, $\Phi_{\max}$ is a PAM on~$[0,1]^d$. By almost same argument, we can see that $\Phi_{\min}$ is also a piecewise affine map on the domain $[0,1]^d$.

In this paper, we investigate the specialisation of the piecewise affine reachability problem to Bellman operators.
\begin{mdframed}
\textbf{Problem BOR (Bellman Operator Reachability).}
Let $\bm{s}, \bm{t} \in [0,1]^d \cap \QQ^d$ and $\Phi: [0,1]^d \rightarrow [0,1]^d$ be a Bellman operator of an MDP~$\mdp$ with no end components. 
Does there exist~$n \in \NN$ such that $\Phi^n(\bm{s}) = \bm{t}?$
\end{mdframed}
We refer to the version of this problem with $\Phi = \Phi_{\max}$ as $\max$-BOR. Similarly, in $\min$-BOR we consider  $\Phi = \Phi_{\min}$.

\section{Bellman Operator Reachability in Arbitrary Dimension}
\label{sec:efs}
In this section, we will discuss the BOR problem without restricting the dimension~$d$.
Recall that for the BOR problem, a Bellman operator~$\Phi \in \{\Phi_{\max}, \Phi_{\min}\}$ has a unique fixed point~$\mu\Phi$.
It will be instrumental to split the discussion of decidability depending on
how the initial and target vectors $\bm{s}$ and $\bm{t}$ compare to~$\mu\Phi$ with respect to the componentwise order on~$\QQ^d$.

We first discuss the properties of the sequence of the iterates
\[\bm{s}, \Phi(\bm{s}), \Phi^2(\bm{s}), \dots\]
which are common for all~$\bm{s} \in [0,1]^d$. 
Based on them, we show that BOR is decidable when $\bm{t}$ is not the unique fixed point~$\mu\Phi$. 
\begin{lemma}\label{Phi-lip-like}
	Let $\Phi \in \{\Phi_{\max},\Phi_{\min}\}$ be a Bellman operator~$\Phi :[0, 1]^d\rightarrow [0, 1]^d$.
	Consider an arbitrary vector~$\bm{x} \in [0,1]^d$ and let
	$\delta:= ||\bm{x} - \mu\Phi||$.
	We have $||\Phi(\bm{x}) - \mu\Phi|| \leq \delta$.
\end{lemma}
\begin{proof}
	For both $\Phi = \Phi_{\max}$ and $\Phi = \Phi_{\min}$ we need to prove the inequalities
	\begin{equation} \label{eq:Phi-not-too-high}
		\Phi(\bm{x})_i \leq \mu\Phi_i + \delta
	\end{equation} and 
\begin{equation}\label{eq:Phi-not-too-low}
	\Phi(\bm{x})_i \geq \mu\Phi_i - \delta
\end{equation}
for all $1\leq i \leq d$ under the assumption $||\bm{x} - \mu\Phi|| = \delta$.
In particular, $\bm{x} \leq \mu\Phi + \bm{\delta}$, where $\bm{\delta} := (\delta, \dots, \delta)$.
For any action~$\alpha$ in state~$i$, we have
\begin{align}\label{eq:act-contract}
	L_{\alpha} (\mu\Phi + \bm{\delta})={} 
	&\sum_{j \in S_d} \prob(i,\alpha, j) \cdot (\mu\Phi_j + \delta) + \prob(i, \alpha, t)\nonumber \\
	={}&\sum_{j \in S_d} \prob(i,\alpha, j) \cdot \mu\Phi_j + \prob(i, \alpha, t) + \sum_{j \in S_d} \prob(i,\alpha, j) \cdot \delta  \nonumber \\
	={}&L_{\alpha}(\mu\Phi) +  \sum_{j \in S_d} \prob(i,\alpha, j) \cdot \delta \leq 
	L_{\alpha}(\mu\Phi) + \delta.
\end{align}

We first consider the operator~$\Phi_{\max}$.
Let \[\tilde{\alpha} := \argmax_{\alpha \in Act_i} L_\alpha(\mu\Phi_{\max} + \bm{\delta}).\] 
By~\eqref{eq:act-contract}, $(\Phi_{\max}(\mu\Phi_{\max} + \bm{\delta}))_i = L_{\tilde{\alpha}}(\mu\Phi_{\max} + \bm{\delta}) \leq L_{\tilde{\alpha}}(\mu\Phi_{\max}) + \delta \leq (\mu\Phi_{\max})_i + \delta$.
From the monotonicity of~$\Phi_{\max}$, we conclude \[(\Phi_{\max}(\bm{x}))_i \leq (\Phi_{\max}(\mu\Phi_{\max} + \bm{\delta}))_i \leq (\mu\Phi_{\max})_i + \delta.\]

Now, for the operator~$\Phi_{\min}$ consider an action~$\alpha'$ such that $L_{\alpha'}(\mu\Phi_{\min}) = (\mu\Phi_{\min})_i$. 
Then
\[(\Phi_{\min}(\mu\Phi_{\min}+\bm{\delta}))_i \leq L_{\alpha'} (\mu\Phi_{\min} + \bm{\delta}) \leq
L_{\alpha'}(\mu\Phi_{\min}) + \delta = (\mu\Phi_{\min})_i + \delta,\]
where the second inequality follows from~\eqref{eq:act-contract}.
By monotonicity, \[(\Phi_{\min}(\bm{x}))_i \leq (\Phi_{\min}(\mu\Phi_{\min}+\bm{\delta}))_i \leq (\mu\Phi_{\min})_i + \delta.\]
This shows~\cref{eq:Phi-not-too-high} for both $\Phi = \Phi_{\max}$ and $\Phi = \Phi_{\min}$.

Now we are going to use~$\mu\Phi - \bm{\delta} \leq \bm{x}$ in order to prove~\cref{eq:Phi-not-too-low}. We note
\begin{align}\label{eq:act-contract-2}
	L_{\alpha} (\mu\Phi - \bm{\delta})={} 
	&\sum_{j \in S_d} \prob(i,\alpha, j) \cdot (\mu\Phi_j - \delta) + \prob(i, \alpha, t)\nonumber \\
	={}&\sum_{j \in S_d} \prob(i,\alpha, j) \cdot \mu\Phi_j + \prob(i, \alpha, t) - \sum_{j \in S_d} \prob(i,\alpha, j) \cdot \delta \nonumber  \\
	={}&L_{\alpha}(\mu\Phi) - \sum_{j \in S_d} \prob(i,\alpha, j) \cdot \delta \geq 
	L_{\alpha}(\mu\Phi) - \delta.
\end{align}

We consider~$\Phi_{\max}$ first. Take
some action~$\alpha'$ in state~$i$ such that $L_{\alpha'}(\mu\Phi_{\max}) = (\mu\Phi_{\max})_i$. By~\eqref{eq:act-contract-2} we have
\[(\Phi_{\max}(\mu\Phi_{\max}-\bm{\delta}))_i \geq L_{\alpha'} (\mu\Phi_{\max} - \bm{\delta}) \geq
L_{\alpha'}(\mu\Phi_{\max}) - \delta = (\mu\Phi_{\max})_i - \delta.\]
By the monotonicity, $(\Phi_{\max}(\bm{x}))_i \geq (\Phi_{\max}(\mu\Phi_{\max}-\bm{\delta}))_i \geq (\mu\Phi_{\max})_i - \delta$.

For~$\Phi_{\min}$ let \[\tilde{\alpha} := \argmin_{\alpha \in Act_i} L_\alpha(\mu\Phi_{\min}  - \bm{\delta}).\] 
We get
\begin{align*}
	(\Phi_{\min}(\bm{x}))_i &\geq (\Phi_{\min}(\mu\Phi_{\min} - \bm{\delta}))_i = L_{\tilde{\alpha}}(\mu\Phi_{\min}  - \bm{\delta}) \\
&\geq L_{\tilde{\alpha}}(\mu\Phi_{\min}) - \delta \geq (\mu\Phi_{\min})_i - \delta.
\end{align*}
This shows~\cref{eq:Phi-not-too-low} for both~$\Phi \in \{\Phi_{\max},\Phi_{\min}\}$ and concludes the proof.
\end{proof}

\begin{proposition}\label{t-not-fp}
	Let $\Phi \in \{\Phi_{\max}, \Phi_{\min}\}$
	be a Bellman operator
	$\Phi : [0, 1]^d\rightarrow [0, 1]^d$
	of an MDP with no end components, and $\bm{s} \in [0,1]^d$ be an arbitrary initial vector.
	For every $\bm{t} \in [0,1]^d$ with $\bm{t} \neq \mu\Phi$, 
	there exists an effectively computable bound~$N$ such that
	\[\Phi^n(\bm{s}) = \bm{t} \quad \Rightarrow \quad \Phi^n(\bm{s}) = \bm{t} \text{ for some } n\leq N.\]
\end{proposition}
\begin{proof}
	Fix vectors~$\bm{s}, \bm{t} \in [0,1]^d$ assuming $\bm{t} \neq \mu\Phi$, where $\mu\Phi$ is the unique fixed point of $\Phi$.
	We crucially use the convergence properties of the interval iteration algorithm~\cite{HaddadM18}.
	Let $\bm{1}:= (1, \dots, 1)$ be the greatest element of the lattice~$[0, 1]^d$. 
	Choose a convergence threshold~$\varepsilon > 0$.
	From \cite[Theorem~2]{HaddadM18} we have $||\Phi^N(\bm{0}) - \Phi^N(\bm{1})|| < \varepsilon$ for some $N \leq A \lceil \frac{\log{\varepsilon}}{\log(1-B^A)} \rceil$,
	where constants~$A, B$ only depend on the MDP and can be computed directly from its representation.
	
	From the monotonicity of $\Phi$ we have $\Phi^N(\bm{0}) \leq \Phi^N(\bm{s}) \leq \Phi^N(\bm{1})$ for any~$\bm{s} \in [0,1]^d$.
	We also have $\Phi^N(\bm{0}) \leq \mu\Phi \leq \Phi^N(\bm{1})$, 
	and so %
	$||\Phi^N(\bm{s}) - \mu\Phi|| < \varepsilon$.

Recall that we can compute the vector~$\mu\Phi$ exactly.
Now choose the threshold $\varepsilon := ||\bm{t} - \mu\Phi||$, and let $N_\varepsilon$ be the previously discussed bound for this threshold.
We have $||\Phi^{N_\varepsilon}(\bm{s}) - \mu\Phi|| <  ||\bm{t} - \mu\Phi||$  and hence, from~\cref{Phi-lip-like},
\[||\Phi^n(\bm{s}) - \mu\Phi|| <  ||\bm{t} - \mu\Phi||\]
for every~$n \geq N_\varepsilon$.
Therefore, either $\bm{t} = \Phi^n(\bm{s})$ for some $n < N_\varepsilon$, or $\bm{t}$ is not reachable from $\bm{s}$ under iteratively applying~$\Phi$.
The first condition can be checked in finite time %
since $N_{\varepsilon}$ is effectively bounded.
\end{proof}

From \cref{t-not-fp} we immediately derive an algorithmic procedure for the BOR problem instances $(\Phi, \bm{s}, \bm{t})$, where $\bm{t} \neq \mu\Phi$. 
For an instance like this, it suffices to compute the bound~$N$ as above, and to test whether $\Phi^n(\bm{s})$ is equal to $\bm{t}$ for some~$n\leq N$ .

We now move on to the case $\bm{t} = \mu\Phi$.
In the sequel, it will be important to differentiate between two types of actions. These types are defined based on preserving the probabilities of the unique fixed point~$\bm{t} = (t_1, \dots, t_d)$.

\begin{definition}
Let $\Phi \in \{\Phi_{\max},\Phi_{\min}\}$ and let $(t_1, \dots, t_d) = \mu\Phi$. 
An action~$\alpha$ available in state~$i$ is \emph{tight}, if $t_i = L_\alpha(t_1, \dots, t_d)$, where $L_\alpha$ is the linear polynomial of action $\alpha$.
An action $\alpha$ is \emph{leaking} in state~$i$, if it is not tight.
\end{definition}

We consider four scenarios next. Along with the distinction between $\Phi= \Phi_{\max}$ and $\Phi=\Phi_{\min}$, we look into the assumptions $\bm{s} \geq \bm{t}$ and $\bm{s} \leq \bm{t}$.

The notion of the sign abstraction will be instrumental in the discussion of the comparable case.
Let $f:[0,1]^d \rightarrow \{-1,0,1\}^d$ be defined by associating 
a sign vector~$\eps = f(\bm{x}) = (\varepsilon_1, \dots, \varepsilon_d)$ 
with every vector~$\bm{x} \in [0,1]^d$, where 
\[\varepsilon_i = \begin{cases}
	1, & x_i > t_i,\\
	0, & x_i = t_i,\\
	-1, & x_i < t_i.
\end{cases}\]
According to this definition, $f(\bm{t}) = \bm{0}$.

\subsection{The cases $\bm{s} \leq \mu\Phi_{\max}$ and $\mu\Phi_{\min} \leq \bm{s} $} \label{sec:below}
In this subsection, we consider two sequences of iterates
\begin{equation}
	\label{max-below}
	\bm{s}, \Phi_{\max}(\bm{s}), \Phi^2_{\max}(\bm{s}), \; \dots, \qquad
	\text{where } \bm{s} \leq \mu\Phi_{\max};
\end{equation}
\begin{equation}
	\label{min-above}
	\bm{s}, \Phi_{\min}(\bm{s}), \Phi^2_{\min}(\bm{s}), \; \dots, \qquad 
	\text{where } \mu\Phi_{\min} \leq \bm{s}.
\end{equation}
Note that by monotonicity,
\begin{enumerate}
	\item If $\bm{s} \leq \mu\Phi_{\max}$, then $\Phi^n_{\max}(\bm{s}) \leq \mu\Phi_{\max}$ for all~$n \geq 0$;
	\item If $\mu\Phi_{\min} \leq \bm{s}$, then $\mu\Phi_{\min} \leq \Phi^n_{\min}(\bm{s})$ for all~$n \geq 0$.
\end{enumerate}

The scenarios outlined in~\cref{max-below,min-above} are discussed together 
--- for them, the sign abstraction~$f$ is well-defined.
Intuitively, we can abstract away from the actual probabilities in the vectors $\Phi^n(\bm{s})$, $n \geq 0$, where $\Phi \in \{\Phi_{\min}, \Phi_{\max}\}$. This succeeds by \emph{only keeping track of whether these probabilities are different from the probabilities in~$\bm{t}$}.
As a result, the space of the BOR problem becomes finite.

We now prove formally that the successor of $f(\bm{x})$ with respect to the Bellman operator is well-defined in the two cases we discuss here.
This would allow to only consider the evolution of the sign vectors later on.

\begin{lemma}\label{phi-below}
Let~$\bm{x}$ and $\bm{y}$ be two vectors satisfying $\bm{x} \leq \bm{t}$ and $\bm{y} \leq \bm{t}$. 
Provided $f(\bm{x}) = f(\bm{y})$, we have
\(f(\Phi_{\max}(\bm{x})) = f(\Phi_{\max}(\bm{y})).\)
\end{lemma}
\begin{proof}
	We consider $\bm{z} \in \{\bm{x},\bm{y}\}$ and let $\eps = (\varepsilon_1, \dots, \varepsilon_d)$ be the sign abstraction of~$\bm{z}$, that is, $\eps = f(\bm{z})$.
	Further, let $\eps' = (\varepsilon'_1, \dots, \varepsilon'_d)$ be the abstraction of~$\Phi_{\max}(\bm{z})$, that is, $\eps' = f(\Phi_{\max}(\bm{z}))$.
	
	Observe first that for a leaking action~$\beta$ chosen in state~$i$ at~$\bm{z}$, we have $L_\beta(\bm{z}) \leq L_\beta(\bm{t}) < t_i$.
Second, let~$\alpha$ be any action in state~$i$. If there exists $j \in \suc(\alpha)$ such that $\varepsilon_j < 0$,
then 
\[L_\alpha(\bm{z}) \leq L_{\alpha}(t_1, \dots, t_{j-1}, z_j, t_{j+1}, \dots, t_d) < L_{\alpha}(\bm{t}) \leq t_i.\]
Therefore, if  $\varepsilon_j < 0$ for some $j \in \suc(\alpha)$, then
$\epsilon'_i < 0$.
Putting it together, we have $\Phi_{\max}(\bm{z})_i = t_i$ if and only if
there exists a tight action~$\alpha \in Act_i$ 
such that for each~$j \in \suc{\alpha}$, we have $\varepsilon_j = 0$.

We summarise these observations as 
\begin{equation}\label{eq:phi-below}
	\varepsilon'_i = \max_{\substack{\alpha \in Act_i \\ \alpha \text{ tight}}} \; \min_{j \in \suc(\alpha)} \varepsilon_j,
\end{equation}
and thus make sure that $\varepsilon'_i$ does not depend on the actual values in~$\bm{x}$ and $\bm{y}$, as long as those two vectors have the same sign abstraction.
\end{proof}

\begin{lemma}\label{psi-above}
	Let~$\bm{x}$ and $\bm{y}$ be two vectors satisfying $\bm{t}\leq \bm{x}$ and $\bm{t}\leq \bm{y}$. 
	Provided $f(\bm{x}) = f(\bm{y})$, we have
	\(f(\Phi_{\min}(\bm{x})) = f(\Phi_{\min}(\bm{y})).\)
\end{lemma}
\begin{proof}
	The proof is analogous to that of~\cref{phi-below}.
	For a vector $\bm{z} \in \{\bm{x}, \bm{y}\}$ and its sign abstraction $\eps = f(\bm{z})$, 
	we note that taking a leaking action~$\beta$ implies $L_\beta(\bm{z}) > t_i$.
	Furthermore, for any action~$\alpha$,
	if  $\varepsilon_j > 0$ for some $j \in \suc(\alpha)$, then
	$\epsilon'_i > 0$.
	This yields	
	\begin{equation}\label{eq:psi-above}
		\varepsilon'_i = \min_{\substack{\alpha \in Act_i \\ \alpha \text{ tight}}} \; \max_{j \in \suc(\alpha)} \varepsilon_j,
	\end{equation}	
or $\Phi_{\min}(\bm{z})_i = t_i$ if and only if
there exists a tight action~$\alpha \in Act_i$ 
such that for each~$j \in \suc{\alpha}$, we have $\epsilon_j = 0$.
\end{proof}

\begin{proposition}\label{below-decid}
There exists an algorithmic procedure for the $\max$-BOR problem instances $(\Phi_{\max}, \bm{s}, \bm{t})$ with $\bm{s} \leq \mu\Phi_{\max}$ and $\bm{t} = \mu\Phi_{\max}$,
as well as for the $\min$-BOR problem instances
 $(\Phi_{\min}, \bm{s}, \bm{t})$ with $\mu\Phi_{\min}\leq \bm{s}$ and $\bm{t} = \mu\Phi_{\min}$.
\end{proposition}
\begin{proof}
It is easy to observe that the space of possible sign vectors is finite.
If we consider $\bm{s} \leq \mu\Phi_{\max}$, then the space is $\{-1,0\}^d$, whereas for the case $\mu\Phi_{\min} \leq \bm{s}$, it is $\{0,1\}^d$.

Given~$\bm{s}$, we compute the abstraction~$f(\bm{s})$ and ask whether $\bm{0}$ is reached by iteratively applying the map $\eps \mapsto \eps'$ defined by~\eqref{eq:phi-below} and~\eqref{eq:psi-above}, respectively.
Once an already explored vector occurs in the sequence $\langle\eps = f(\bm{s}), \eps', \eps'', \dots \rangle$, 
we can stop.
This happens in at most~$2^d -1$ iterations.
In this finite sequence, $\bm{0}$ occurs if and only if~$\bm{t}$ is reached by iterating~$\Phi \in \{\Phi_{\max}, \Phi_{\min}\}$ starting from~$\bm{s}$.
This is due to~\cref{phi-below,psi-above}.
\end{proof}

\subsection{The cases $\mu\Phi_{\max} \leq \bm{s}$ and $\bm{s} \leq \mu\Phi_{\min}$} \label{sec:above}

In this subsection, we consider two sequences of iterates
\begin{equation}
	\label{max-above}
		\bm{s}, \Phi_{\max}(\bm{s}), \Phi^2_{\max}(\bm{s}), \; \dots, \qquad
	\text{where } \mu\Phi_{\max} \leq \bm{s};
\end{equation}
\begin{equation}
	\label{min-below}
		\bm{s}, \Phi_{\min}(\bm{s}), \Phi^2_{\min}(\bm{s}), \; \dots, \qquad 
	\text{where } \bm{s} \leq \mu\Phi_{\min}.
\end{equation}
By monotonicity,
\begin{enumerate}
	\item If $\mu\Phi_{\max} \leq \bm{s}$, then $\mu\Phi_{\max} \leq \Phi^n_{\max}(\bm{s})$ for all~$n \geq 0$;
	\item If $\bm{s} \leq \mu\Phi_{\min}$, then $\Phi^n_{\min}(\bm{s}) \leq \mu\Phi_{\min}$ for all~$n \geq 0$.
\end{enumerate}

The main result of this subsection is the following proposition.
\begin{proposition}
\label{above-decid}

There exists an algorithmic procedure for the $\max$-BOR problem instances $(\Phi_{\max}, \bm{s}, \bm{t})$ with $ \mu\Phi_{\max} \leq \bm{s}$ and $\bm{t} = \mu\Phi_{\max}$,
as well as for the $\min$-BOR problem instances
$(\Phi_{\min}, \bm{s}, \bm{t})$ with $\bm{s} \leq \mu\Phi_{\min}$ and $\bm{t} = \mu\Phi_{\min}$.
\end{proposition}
The procedure for this case is more intricate than in~\cref{sec:below}.
This is due to the new phenomenon that occurs for the sequences~\eqref{max-above} and \eqref{min-below}.
While the value $f(\Phi(\bm{x}))$ is determined by tight actions in~\cref{sec:below}, it is not the case for the sequences~\eqref{max-above} and \eqref{min-below}. 
An iteration of the Bellman operator can choose a leaking action~$\beta \in Act_i$ over all tight actions available in state~$s_i$. 
Moreover, this can determine the value~$f(\Phi(\bm{x}))$. 
Consider~\eqref{max-above}: 
intuitively, a leaking action~$\beta$ is chosen if the successor states~$\suc(\beta)$ have probabilities significantly greater than optimal---enough to compensate for the ``leakage''~$t_i - L_\beta(\bm{t})$.

\begin{figure}[t]
\begin{minipage}[b]{0.33\columnwidth}
	\begin{center}
		\begin{tikzpicture}[on grid,auto]
			\node[state] (s0) {$\mathstrut s_1$};
			\node[dot] (n0) [below right=0.4 and 1 of s0] {};
			\node[state] (s1) [below right=0 and 1.5 of s0] {$\mathstrut s_2$};
			\node[state] (s2) [below right=0 and 1.5 of s1] {$\mathstrut s_3$};
			\node[state, accepting] (sp) [below left =1.25 and 0 of s1] {$t$};
			\node[] (me) [above left=0.5 and 0.5 of s0] {\small$\mdp_1$:};
			;
			\path[->]
			(s0) edge [bend left] node[pos=0.5,inner sep=1pt] {$\beta$, 1/2} (s1)
			(s0) edge [bend right] node[swap,pos=0.7,inner sep=2pt] {$\alpha$} (n0)
			(n0) edge [bend right] node[pos=0.7,inner sep=1pt] {1/2} (sp)
			(s0) edge [bend right] node[swap,inner sep=1pt] {$\beta$, 5/12} (sp)
			(n0) edge [bend right] node[pos=0.6,inner sep=1pt] {1/3} (s2)
			(s1) edge [bend left] node[pos=0.5,inner sep=1pt] {1} (s2)
			(s2) edge [bend left] node[pos=0.5,inner sep=1pt] {1/4} (sp)
			;
		\end{tikzpicture}
	\end{center}
	\caption{An MDP $\mdp_1$.}
	\label{fig:mdpLeakingAbove}
\end{minipage}
\hfill
\begin{minipage}[b]{0.5\columnwidth}
	\begin{center}
		\begin{tikzpicture}[on grid,auto]
			\node[state] (s1) {$\mathstrut s_1$};
			\node[dot] (n1) [above right= 0 and 1 of s1] {};
			\node[dot] (n2) [above right= 0 and -1 of s1] {};
			\node[state] (s2) [above right = 1 and 2 of s1] {$\mathstrut s_2$};
			\node[state] (s3) [above right = -1 and 2 of s1] {$\mathstrut s_3$};
			\node[state, accepting] (s4) [above right = -1 and 1 of s2] {$\mathstrut t$};
			\node[] (me) [above left=1 and 2.5 of n1] {\small$\mdp_2$:};
			\path[->]
			(s1) edge [bend right] node[swap,pos=0.4,inner sep=1pt] {$\alpha_1$} (n1)
			(s1) edge [bend right] node[swap,pos=0.4,inner sep=1pt] {$\alpha_2$} (n2)
			(n1) edge [bend right] node[swap,pos=0.6,inner sep=1pt] {1/3} (s1)
			(n2) edge [bend right] node[swap,pos=0.6,inner sep=1pt] {1/2} (s1)
			(n1) edge [bend right] node[pos=0.6,inner sep=1pt] {1/3} (s2)
			(n1) edge [bend left] node[swap,pos=0.6,inner sep=1pt] {1/3} (s3)
			(n2) edge [out = 85, in = 140] node[pos=0.6,inner sep=1pt] {1/4} (s2)
			(n2) edge [out = 275, in = 220] node[swap,pos=0.6,inner sep=1pt] {1/4} (s3)
			(s2) edge [bend left] node[pos=0.6,inner sep=1pt] {$\beta$, 1/6} (s4)
			(s3) edge [bend right] node[swap,pos=0.5,inner sep=1pt] {$\gamma$, 1/6} (s4)
			(s2) edge [bend right] node[swap,pos=0.6,inner sep=1pt] {$\beta$, 1/3} (s1)
			(s3) edge [bend left] node[pos=0.6,inner sep=1pt] {$\gamma$, 1/3} (s1)
			(s2) edge [loop,out=60,in=-0,looseness=4] node {$\beta$, 1/3} (s2)
			(s3) edge [loop,out=0,in=-60,looseness=4] node {$\gamma$, 1/3} (s3)
			;
		\end{tikzpicture}
		\caption{An MDP $\mdp_2$.}
		\label{fig:mdpNonPositional}
	\end{center}
\end{minipage}
\end{figure}
\begin{example}
Consider the MDP $\mdp_1$ in~\cref{fig:mdpLeakingAbove}. 
It has~$S = \{s_1, s_2, s_3, s_{-}, t\}$ where the ``missing'' probabilistic transitions lead to~$s_{-}$.
We omit~$s_{-} \not\in S_d$ and transitions to it, for simplicity of presentation.

With respect to the max-operator~$\Phi_{\max}$ and $\bm{t} = \mu\Phi_{\max}$, there is one tight action $\alpha$ and one leaking action $\beta$ in~$s_1$. 
This example shows that it is possible to have $L_\beta(\bm{x}) > t_1$ whereas $L_{\alpha}(\bm{x}) \leq t_1$ holds for all tight~$\alpha \in Act_1$.
This is different from the case of $\Phi_{\max}$ and $\bm{x} \leq \bm{t}$ --- cf.\ \cref{eq:phi-below}.

Let $\bm{s} \coloneq (1, 1/3, 2/3)$. Clearly, $\bm{s} > \mu\Phi_{\max} = (7/12, 1/4, 1/4)$. 
The tight action $\alpha$ is chosen for the first iteration, and $\Phi_{\max}(\bm{s}) = \left(13/18, 2/3, 1/4\right)$. 
Observe how $L_\alpha(\Phi_{\max}(\bm{s})) = t_1$ but $L_\beta(\Phi_{\max}(\bm{s})) > t_1$, and hence the latter is chosen by~$\Phi_{\max}$ in this step. 
That is, the leaking action $\beta$ is chosen and $\Phi_{\max}^{2}(\bm{s}) = \left(9/12, 1/4, 1/4\right)$.
Finally, $\Phi_{\max}^{3}(\bm{s}) = \bm{t} = \mu\Phi_{\max}$ by choosing the tight action $\alpha$.
\end{example}

\subparagraph*{Only tight actions eventually.} 
However, we show that in the convergent sequences~\eqref{max-above} and~\eqref{min-below}, the actions chosen by the respective Bellman operators~$\Phi_{\max}$ and~$\Phi_{\min}$ are all tight, after some number of iterations.

\begin{lemma}\label{nbhd-tight}
Let $\Phi \in \{\Phi_{\max}, \Phi_{\min}\}$ be a Bellman operator $\Phi: [0, 1]^d\rightarrow [0, 1]^d$ such that $\mu\Phi = \bm{t}$.
There exists a $\delta$-neighbourhood of the fixed point
\[U_\delta(\bm{t}) = \{\bm{x} \in [0,1]^d : d(\bm{x}, \bm{t}) < \delta \}\]
such that 
for every $\bm{x} \in U_\delta(\bm{t})$, the vector $\Phi(\bm{x})$ is obtained by applying only tight actions.
That is, 
\[\Phi(\bm{x}) = \left(L_{\alpha_1}(\bm{x}), \dots, L_{\alpha_d}(\bm{x})\right),\]
where each $\alpha_i \in Act_i$, $1\leq i\leq d$, is tight.
\end{lemma}
\begin{proof}
Since the transition probabilities in~$\mathcal{M}$ are rational numbers, one argues that $t_1, \dots, t_d$ are rational, too.
We consider the set of rational numbers that consists of
$t_1, \dots, t_d$, along with $L_\alpha(\bm{t})$ for each action~$\alpha$.
It is sufficient to consider leaking actions because for a tight action~$\alpha$ 
applied in state~$i$, $L_\alpha(\bm{t}) = t_i$.
Let~$D$ be the least common denominator of the numbers in the aforedescribed set.

Let $\delta:= \frac{1}{2D}$ and pick~$\bm{x} \in U_\delta(\bm{t})$.
Further let $A = (\alpha_1, \dots, \alpha_d)$ be the actions chosen at~$\bm{x}$ by the Bellman operator~$\Phi$.
We show that each action in~$A$ is tight.
Denote by~$f_A$ the effect of applying~$A$, i.e.,
\[f_A(\bm{t}) = \left(L_{\alpha_1} (\bm{t}), \dots, L_{\alpha_d} (\bm{t})\right) \quad \text{and} \quad
f_A(\bm{x}) = \Phi(\bm{x}).
\]

We notice that~$f_A$ is 1-Lipschitz (which, in fact, holds for any choice~$A$ of actions). 
Indeed, let us consider arbitrary~$\bm{u}, \bm{v} \in [0,1]^d$. 
We have \begin{align*}
	||f_A(\bm{u}) - f_A(\bm{v})|| &= \max_{1 \leq i \leq d} |L_{\alpha_i}(\bm{u}) - L_{\alpha_i}(\bm{v})|\\
	&\leq\max_{1\leq i \leq d} \sum_{j \in \suc(\alpha_i)} \prob(i, \alpha_i, j)\cdot  \left| u_j-v_j \right|\\
	&\leq \max_{1\leq i \leq d} \left(1 \cdot \max_{j \in \suc{\alpha_i}}|u_j - v_j|\right)\leq
	\max_{1\leq i \leq d} \left| u_i-v_i \right| = ||\bm{u} - \bm{v}||.
\end{align*}
Utilising the 1-Lipschitz property, we observe $||f_A(\bm{t}) - \Phi(\bm{x})|| \leq ||\bm{t} - \bm{x}||$.
Moreover, we have $||\Phi(\bm{x}) - \bm{t}|| \leq ||\bm{x} - \bm{t}||$ due to~\cref{Phi-lip-like}. 
Therefore, 
\begin{align}\label{eq:A-optimal}
	||f_A(\bm{t}) - \bm{t}|| &\leq ||f_A(\bm{t}) - \Phi(\bm{x})|| + ||\Phi(\bm{x}) - \bm{t}||\nonumber\\
	&\leq ||\bm{t} - \bm{x}|| + ||\bm{x} - \bm{t}|| < \delta + \delta = \frac{1}{D}. 
\end{align} 

On the other hand, $
||f_A(\bm{t}) - \bm{t}|| = \max_{1 \leq i \leq d} |L_{\alpha_i}(\bm{t}) - t_i|$.
By the definition of~$D$, we know that $L_{\alpha_i}(\bm{t}) \neq t_i$ implies $|L_{\alpha_i}(\bm{t}) - t_i| \geq \frac{1}{D}$.
We conclude from~\cref{eq:A-optimal} that $f_A(\bm{t}) = \bm{t}$.
Equivalently, all actions chosen at~$\bm{x}$ are tight.
\end{proof}

The vector sequence $\bm{s}, \Phi(\bm{s}), \Phi^2(\bm{s}), \dots$ converges to~$\bm{t}$ due to~\cref{all-converge}. Therefore, it reaches a $\frac{1}{2D}$-neighbourhood of~$\bm{t}$ after finitely many steps.
Furthermore, an upper bound on the number of necessary steps can be computed as in~\cref{t-not-fp}.
\begin{corollary}\label{event-tight}
Let $\Phi \in \{\Phi_{\max},\Phi_{\min}\}$.
For an arbitrary initial vector $\bm{s} \in [0,1]^d$, there exists an effectively computable~$N \in \NN$ such that in the sequence $(\Phi^n(\bm{s}))_{n\in \NN}$
for every $n \geq N$, $\Phi^{n+1}(\bm{s})$ is obtained by applying only tight actions to $\Phi^n(\bm{s})$.
\end{corollary}
We point out that the argument used in the proof of~\cref{nbhd-tight},
bears resemblance to (and is inspired by) the proof of~\cite[Theorem~3]{HaddadM18}.
However, the eventual optimality of tight actions,
which we establish here,
is not a matter of discussion in~\cite{HaddadM18} or any other work we know.

\subparagraph*{From probabilities to $\{-1,0,1\}$.} 
As before, we show that the successor of~$f(\bm{x})$ is well-defined with respect to the Bellman operator.
For this case, it holds in a certain neighbourhood of the fixed point~$\bm{t}$. 
\begin{lemma}\label{abstr-with-tight}
Let~$\Phi \in \{\Phi_{\max},\Phi_{\min}\}$ and let~$\delta$ be chosen as above
to guarantee that~$\Phi$ only picks tight actions in the neighbourhood~$U_\delta(\bm{t})$.

Consider vectors~$\bm{x},\bm{y} \in U_\delta(\bm{t})$, both comparable with~$\bm{t}$.
Provided $f(\bm{x}) = f(\bm{y})$, we have
\(f(\Phi(\bm{x})) = f(\Phi(\bm{y})).\)
\end{lemma}
\begin{proof}
The assumption~$\bm{x}, \bm{y} \in U_\delta(\bm{t})$ is necessary to shake off the effect of the leaking actions (cf.\ the assumptions of~\cref{phi-below}).
Since we assume $f(\bm{x}) = f(\bm{y})$, 
	note that $\bm{x} \leq \bm{t}$ is equivalent to~$\bm{y} \leq \bm{t}$,
	and $\bm{t} \leq \bm{x}$ is equivalent to $\bm{t} \leq \bm{y}$.

Let $\bm{z} = (z_1, \dots, z_d)$ be an arbitrary vector in~$U_\delta(\bm{t})$.
Recall that this guarantees that for every~$i$, there is a tight action~$\alpha \in Act_i$ with 
$\Phi(\bm{z})_i = L_\alpha(\bm{z})$.

First, consider $\Phi = \Phi_{\max}$. We can assume $\bm{z} \geq \bm{t}$. 
It is not hard to see now that
having $j \in \suc(\alpha)$ with $z_j > t_j$ implies $L_\alpha(\bm{z}) > t_i$.
Therefore, $(\Phi_{\max}(\bm{z}))_i = t_i$ if and only if there exists no tight action~$\alpha \in Act_i$ that depends on
$j \in \suc(\alpha)$ with $z_j > t_j$.
Equivalently, we have
\begin{equation}\label{eq:phi-above}
	\varepsilon'_i = \max_{\substack{\alpha \in Act_i \\ \alpha \text{ tight}}} \; \max_{j \in \suc(\alpha)} \varepsilon_j,
\end{equation}
where~$\eps' = (\varepsilon'_1, \dots, \varepsilon'_d)$ is the abstraction of~$\Phi_{\max}(\bm{x})$, that is, $\eps' = f(\Phi_{\max}(\bm{x}))$.

Now consider $\Phi = \Phi_{\min}$ and assume $\bm{z} \leq \bm{t}$.
	Having $j \in \suc(\alpha)$ with $z_j < t_j$ implies $L_\alpha(\bm{z}) < t_i$.
Similarly to the $\max$-case, we have $(\Phi_{\min}(\bm{z}))_i = t_i$ if and only if there is no tight action~$\alpha \in Act_i$ that depends on
$j \in \suc(\alpha)$ with $z_j < t_j$.
Equivalently, we have
\begin{equation}\label{eq:psi-below}
	\varepsilon'_i = \min_{\substack{\alpha \in Act_i \\ \alpha \text{ tight}}} \; \min_{j \in \suc(\alpha)} \varepsilon_j,
\end{equation}
where~$\eps' = (\varepsilon'_1, \dots, \varepsilon'_d)$ is the abstraction of~$\Phi_{\min}(\bm{x})$, that is, $\eps' = f(\Phi_{\min}(\bm{x}))$.
	
Therefore, the abstraction vectors of~$\Phi(\bm{x})$ and~$\Phi(\bm{y})$ do not depend on the actual values in~$\bm{x}$ and~$\bm{y}$, but only on their abstractions $f(\bm{x}) = f(\bm{y})$.
\end{proof}
\begin{proof}[Proof of~\cref{above-decid}]
We consider both~\cref{max-above,min-below}, so we write~$\Phi \in \{\Phi_{\max},\Phi_{\min}\}$ and~$\bm{t} = \mu\Phi$.
	First, compute the bound~$N$ of~\cref{event-tight}.
If $\Phi^n(\bm{s}) = \bm{t}$ holds for some $n < N$, we terminate with a positive answer to the BOR problem.
Otherwise, we continue with the abstraction argument.
We compute~$f(\Phi^N(\bm{s}))$ and ask whether $\bm{0}$ is reached from it by iteratively applying the map $\eps \mapsto \eps'$ as defined by~\cref{eq:phi-above} and by~\cref{eq:psi-below}, respectively. 
In at most~$2^d - 1$ iterations we either reach~$\bm{0}$,
or discover an ever-repeating vector subsequence that does not contain~$\bm{0}$.
The former is the case if and only if $\bm{t}$ is reached by iterating~$\Phi$ from~$\Phi^N(\bm{s})$. This is due to~\cref{abstr-with-tight}.
\end{proof}
Combining \cref{below-decid,above-decid}%
, we obtain the following theorem.
\begin{theorem}\label{compar-decid}
There exists an algorithmic procedure that solves all BOR problem instances $(\Phi, \bm{s}, \bm{t})$ with
$\Phi \in \{\Phi_{\max},\Phi_{\min}\}$,
$\bm{t} = \mu\Phi$
and %
\(\bm{s} \; \{\leq, \geq\} \; \mu\Phi.\)
\end{theorem}
\subsection{Initial and target vectors are incomparable}
\label{sec:incomp}
As before, let $\Phi \in \{\Phi_{\max}, \Phi_{\min}\}$ be a Bellman operator with unique fixed point~$\mu\Phi$.
We keep assuming~$\bm{t}= \mu\Phi$ and consider the remaining case, that is, the case when $\bm{s}$ and $\bm{t}$ are two incomparable vectors,
denoted $\bm{s} \bowtie \bm{t}$.

We can assume $\bm{s} \in U_\delta(\bm{t})$ as defined in~\cref{nbhd-tight}.
Clearly, starting with an arbitrary incomparable vector, we can apply~$\Phi$ up to the pre-computed power~$N$,
reaching either a comparable vector (potentially including~$\bm{t}$ itself), or the $\delta$-neighbourhood of~$\bm{t}$.
In the latter case, set $\bm{s}$ to be the first vector inside the neighbourhood.
Hence, for every $\bm{x}$ discussed below, the vector $\Phi(\bm{x})$ is obtained by applying only tight actions.
\subparagraph*{Positionality.}
Even after eventually adhering to tight actions, the behaviour of $\Phi$'s iterations is not described by a single linear transformation.
We recall from~\cref{vi-positional} that positionality is prominently sufficient for achieving optimal probabilities \emph{in the limit}~\cite{bk}.
However, \cref{ex:3d} shows a more subtle behaviour for our reachability problem.
\begin{example}\label{ex:3d}
In~\cref{fig:mdpNonPositional}, we present an MDP~$\mdp_2$ with~$S = \{s_1, s_2, s_3, s_{-}, t\}$ where the ``missing'' probabilistic transitions lead to~$s_{-} \not \in S_d$.
We omit~$s_{-}$ and transitions to it, for simplicity.
Let $Act_1 = \{\alpha_1, \alpha_2\}$, $Act_2 = \{\beta\}$, $Act_3 = \{\gamma\}$ be the actions available in the states of~$S_3$.

For the $\max$-operator $\Phi_{\max}$ observe that $\mu\Phi_{\max} = (\frac{1}{2},\frac{1}{2},\frac{1}{2})$ and all actions are tight, including both~$\alpha_1,\alpha_2$.

Let $\bm{s} = (0,\frac{5}{6},\frac{5}{6})$.
When applying~$\Phi_{\max}$, the action $\alpha_1$ is chosen over~$\alpha_2$ in~$\bm{s}$, and so
$\Phi_{\max}(\bm{s}) = (\frac{5}{9}, \frac{4}{9}, \frac{4}{9})$.
However, in the next iteration, $\alpha_2(\frac{5}{9}, \frac{4}{9}, \frac{4}{9}) = (\frac{1}{2},\frac{1}{2},\frac{1}{2}) \geq 
(\frac{13}{27},\frac{1}{2},\frac{1}{2})=  \alpha_1(\frac{5}{9}, \frac{4}{9}, \frac{4}{9})$ yielding
$\Phi^2_{\max}(\bm{s}) = \mu\Phi_{\max}$.

Spectacularly, none of two positional schedulers reaches~$\mu\Phi_{\max}$, 
which can be proved using%
~\cite{kannan1986poly}.
\end{example}
With positionality out of question, 
we need to study schedulers that switch actions over time.
\subparagraph*{Matrix semigroups.}
In the sequel, we take a matrix perspective on BOR
by introducing a~$d\times d$-matrix for every tuple of tight actions.
The set of such tuples is finite, and thus we argue that the behaviour of
Bellman operator iterations from~$\bm{s}$ is governed by multiplying the vector~$\bm{s}$ with elements of a semigroup~$\mathcal{S}$ generated by finitely many matrices~$M_1, \dots, M_k$.

We associate with every tight action~$\alpha \in Act_i$ a row vector
\[\left( \prob(s_i, \alpha, s_1), \dots, \prob(s_i, \alpha, s_d) \right)\]
of probabilities for going to the states of~$S_d$,
as well as a scalar $\prob(s_i, \alpha, t)$ for reaching~$t$.

For each state~$s_i$, let~$\mathcal{F}_i$ denote the set of row vectors for tight actions in~$Act_i$.
Notice that all row vectors only have non-negative entries; furthermore, the sum of elements in each row vector is at most~1. %
We will further refer to matrices with all rows satisfying these properties as~\emph{substochastic}.

\begin{definition}
A family $\mathcal{F} \subset \QQ^{d \times d}$ of matrices is called a \emph{product family}
if $\mathcal{F}$ consists of all possible matrices with $i$-th row from~$\mathcal{F}_i$ for all~$i \in \{1, \dots, d\}$.

We further let $\mathcal{S}: = \langle \mathcal{F} \rangle$ 
be the semigroup generated by~$\mathcal{F}$.
\end{definition}
We introduce the map~$f$ by $\bm{x} \mapsto \bm{x} - \mu\Phi$.
\begin{mdframed}
\textbf{The Product Family Reachability (PFR) Problem.}
Let $\mathcal{F} = \{M_1, \dots, M_k\}$ be a product family of substochastic matrices, and $\eps \in [-1,1]^d\cap \QQ^d$.
The \emph{PFR problem} asks, given a map $F: \QQ^d \rightarrow \QQ^d$: Does there exist~$n \geq 0$ such that $F^n(\eps) = \bm{0}$? 

Here, $F$ stands for one of the two maps:
\begin{enumerate}
	\item $F_{\max}(\bm{v}) = \max_{1\leq i \leq k} \left(M_i \cdot\bm{v}\right)$,
	\item $F_{\min}(\bm{v}) = \min_{1\leq i \leq k} \left(M_i \cdot\bm{v}\right)$.
\end{enumerate}
\end{mdframed}

Note in passing that $\max$ (respectively, $\min$) over vectors is taken with respect to the partial order~$\leq$.
The operators~$F_{\max}$ and~$F_{\min}$ are thus well-defined, that is,
for $F \in \{F_{\max}, F_{\min}\}$
there exists a matrix~$M \in \mathcal{F}$ such that $F(\bm{v}) = M\cdot \bm{v}$. 
It holds indeed that $F_{\max}(\bm{v})_i \geq \bm{w} \cdot \bm{v}$ and $F_{\min}(\bm{v})_i \leq \bm{w} \cdot \bm{v}$ for any $\bm{w} \in \mathcal{F}_i$, $1 \leq i \leq d$.
\begin{restatable}{proposition}{probequiv}
\label{probequiv}
Every~$d$-dimensional instance of BOR with $\bm{t} = \mu\Phi$ and $\bm{s} \in U_\delta(\bm{t})$ (thus guaranteeing that only tight actions are used)
is equivalent to a~$d$-dimensional instance of the PFR problem.
\end{restatable}
\begin{proof}
Let~$\mdp$ be an MDP whose Bellman operator~$\Phi \in \{\Phi_{\max}, \Phi_{\min}\}$ has a unique fixed point~$\mu\Phi \in [0,1]^d$.
Fix $\bm{t} = \mu\Phi$ and $\bm{s} \in U_\delta(\bm{t})$, where the neighbourhood~$U_\delta(\bm{t})$ is as in~\cref{nbhd-tight}.
We will introduce a product family of substochastic matrices in $\QQ^{d \times d}$ with a corresponding operator~$F$ and
we will prove that there exists~$n\in \NN$ such that $\Phi^n(\bm{s}) = \bm{t}$
if and only if
$F^n(f(\bm{s})) = \bm{0}$.

Let~$\mathcal{F}_i$ be the set of all row vectors for actions in~$Act_i$ of~$\mdp$, for each $i \in \{1, \dots, d\}$.
Then, a finite product family~$\mathcal{F} = \{M_1, \dots, M_k\} \subset \QQ^{d \times d}$ obtained from these sets only contains substochastic matrices.

Let~$A = (\alpha_1, \dots, \alpha_d)$ be a tuple of actions in~$\mdp$, where $\alpha_i \in Act_i$ for each~$i$.
Consider its corresponding matrix~$M \in \mathcal{F}$.
Formally, $(M)_{i,j} = \prob(s_i, \alpha_i, s_j)$.
Let $L_A(\bm{x})$ denote the vector $\left(L_{\alpha_1}(\bm{x}), \dots, L_{\alpha_d}(\bm{x})\right)^\top$.
By definition of a row vector of an action, we have
$L_A(\bm{x}) = M \cdot \bm{x} +  \prob_A$, where $\prob_A = \left(\prob(s_1, \alpha_1, t), \dots, \prob(s_d, \alpha_d, t) \right)^\top$. 
In particular, $\mu\Phi = M \cdot \mu\Phi + \prob_A$ since every action in~$A$ is tight.
As an intermediate step, we prove the claim below.
\begin{claim*}\label{eps-facts}
	\begin{itemize}The following equalities hold for every~$\bm{x} \in U_\delta(\bm{t})$:
		\item $f(L_A(\bm{x})) = M \cdot f(\bm{x})$,
		\item $f(\Phi_{\max}(\bm{x})) = F_{\max}(f(\bm{x}))$ and $f(\Phi_{\min}(\bm{x})) = F_{\min}(f(\bm{x}))$,
		\item for every~$n \geq 1$, $f(\Phi_{\max}^n(\bm{x})) = F_{\max}^n(f(\bm{x}))$ and $f(\Phi_{\min}^n(\bm{x})) = F_{\min}^n(f(\bm{x}))$.
	\end{itemize}
\end{claim*}
\begin{proof}[Proof (of the Claim)]
	Use linearity of~$L_A$ and the equality $L_A(\mu\Phi) = \mu\Phi$ to obtain
	\begin{align*}M \cdot f(\bm{x}) = M \cdot \left(\bm{x} - \mu\Phi\right) = &M \cdot \bm{x} - M \cdot \mu\Phi =\\ M \cdot \bm{x} - (\mu\Phi - \prob_A) =&  M \cdot \bm{x} -
		\mu\Phi + \prob_A\\
		\text{and } \qquad f(L_A(\bm{x})) = L_A(\bm{x}) - \mu\Phi =&
		M \cdot \bm{x} + \prob_A - \mu\Phi.	
	\end{align*}
	The first part of the claim thus holds.
	The second statement follows from $\Phi_{\max}(\bm{x}) = \max_A L_A(\bm{x})$.
	Indeed, \begin{align*}
		f(\Phi_{\max}(\bm{x})) = f(\max_A L_A(\bm{x})) = \max_A (f(L_A(\bm{x}))) = \max_{M \in \mathcal{F}} M \cdot f(\bm{x}) = F_{\max}(f(\bm{x}))
	\end{align*}
and similarly, $f(\Phi_{\min}(\bm{x}))= F_{\min}(f(\bm{x}))$.
	
	Let $\Phi \in \{\Phi_{\max},\Phi_{\min}\}$. 
	The second item further serves both as the base case (with  argument $\Phi(\bm{x})$) 
	and the induction step (with argument $\Phi^{n+1}(\bm{x})$) to prove the third, final item of the claim.
\end{proof}
Now, set~$\eps\coloneqq f(\bm{s}) = \bm{s} - \mu\Phi$ in the PFR problem. Clearly, $\eps \in [-1, 1]^d \cap \QQ^d$.
The equality $\Phi^n(\bm{s}) = \mu\Phi$ holds if and only if $f(\Phi^n(\bm{s})) = \bm{0}$ if and only if $F^n(\eps) = \bm{0}$.
\end{proof}

\emph{\cref{ex:3d} (revisited).} The MDP $\mdp_2$ in~\cref{fig:mdpNonPositional} has a product family $\mathcal{F} = \{M_1,M_2\}$. The two matrices correspond to the action tuples $A_1 = (\alpha_1, \beta, \gamma)$ and $A_2 = (\alpha_2, \beta, \gamma)$:
\[M_1 = \begin{pmatrix}
1/3 & 1/3 & 1/3 \\
1/3 & 1/3 & 0\\
1/3 & 0 & 1/3
\end{pmatrix},
\hspace{1cm}
M_2 = \begin{pmatrix}
1/2 & 1/4 & 1/4 \\
1/3 & 1/3 & 0\\
1/3 & 0 & 1/3
\end{pmatrix}.
\]
We have $\eps:= f(\bm{s}) = \bm{s} - \mu\Phi_{\max} = (-1/2, 1/3, 1/3)$, and 
\[F_{\max}^2(\eps) = \max_{1 \leq i, j \leq 2}\left(M_jM_i\eps\right) = M_2 M_1\eps = \bm{0}.\]
Indeed, %
$\Phi_{\max}^2(\bm{s}) = \mu\Phi_{\max}$.
Meanwhile, $M_1^n \cdot \eps \neq \bm{0}$ and $M_2^n \cdot \eps \neq \bm{0}$ for all~$n$.

\begin{remark}\label{not-just-reach}
	An instance of the PFR problem with $F = F_{\max}$ is a ``yes'' instance if and only if 
	\begin{equation*}\label{eps-pre-fp}
		\exists n. \enspace	\left(\exists M = M_{i_n} \dots M_{i_1}. \enspace
		M \cdot \eps = \bm{0} \right) \enspace
		\wedge \enspace
		\left(	\forall M' = M_{j_n} \dots M_{j_1}. \enspace
		M' \cdot \eps \leq \bm{0} \right),	
	\end{equation*}
	where $i_1, \dots, i_n, j_1, \dots, j_n \in \{1, \dots, k\}$. 
	Similarly, for $F = F_{\min}$ an instance is ``yes'' if and only if 
	\begin{equation*}
		\exists n. \enspace	\left(\exists M = M_{i_n} \dots M_{i_1}. \enspace
		M \cdot \eps = \bm{0} \right) \enspace
		\wedge \enspace
		\left(	\forall M' = M_{j_n} \dots M_{j_1}. \enspace
		M' \cdot \eps \geq \bm{0} \right).
	\end{equation*}
Finding a semigroup element~$M \in \mathcal{S}$ that satisfies $M \cdot \eps = \bm{0}$,
	together with previously discussed techniques, 
	is sufficient to answer the BOR problem for both $\Phi_{\max}$ and $\Phi_{\min}$.
	If~$M \cdot \eps = M_{i_n} \dots M_{i_1}  \cdot \eps=\bm{0}$,
	then
	\begin{enumerate}
		\item $F_{\max}^n(\eps) \geq \bm{0}$ and hence $\Phi^n_{\max}(\bm{s}) \geq \bm{t}$, and
		\item $F_{\min}^n(\eps) \leq \bm{0}$ and hence $\Phi^n_{\min}(\bm{s}) \leq \bm{t}$.
	\end{enumerate} Then, we can apply the complete algorithm from~\cref{above-decid}. 
	Notice that this does not necessarily imply that $(\Phi, \bm{s}, \bm{t})$ is a positive BOR instance.
\end{remark}

However, deciding whether a matrix~$M \in \mathcal{S}$ with $M\cdot\eps = \bm{0}$ exists is per se an undecidable problem for general matrices~\cite{Bell2008undec}.
There, the so-called \emph{vector reachability problem} for a matrix semigroup~$\mathcal{S} = \langle M_1, \dots, M_k \rangle$ asks for given~$\bm{x}, \bm{y} \in \QQ^d$, whether there exists~$M \in \mathcal{S}$ such that $M\cdot \bm{x} = \bm{y}$.
\subparagraph*{Unique tight actions.}
We conclude this section with the discussion of the PFR Problem restricted to semigroups generated by a single matrix.
This corresponds to the assumption that there is a unique tight action in every state~$s_i \in S_d$.
The decidability for this restriction is not surprising and, in particular, 
follows from the algorithmic procedure of~\cite{kannan1980decid}.

Nevertheless, we employ simple linear algebra techniques
to provide an alternative proof for our version of the problem (importantly, the target vector is zero).
\begin{proposition}\label{position-dec}
	Let $M \in \QQ^d$ be a substochastic matrix and $\eps \in [-1,1]^d\cap \QQ^d$ an arbitrary vector.
	If there exists~$n \geq 0$ such that
	$M^n \cdot \eps = \bm{0}$, then
	there exists such $n \leq d$.
	Whether there exists~$n \geq 0$ such that
	$M^n \cdot \eps = \bm{0}$ can thus be answered algorithmically.
\end{proposition}
\begin{proof}
	Let $\ker M$ denote the kernel of~$M$, the set of all vectors~$\bm{x}$ such that~$M \cdot \bm{x} = \bm{0}$.
	
	Clearly, for all~$n\geq 0$, $\ker M^n \subseteq \ker M^{n+1}$.
	Furthermore, notice that $\ker M^n = \ker M^{n+1}$ implies $\ker M^{n+k} = \ker M^n$ for all~$k \geq 0$.
	Assume, for sake of contradiction, that there exists~$n$ and $k \geq 2$ such that $\ker M^n = \dots = \ker M^{n+k-1} \subsetneq \ker M^{n+k}$.
	Then consider~$\bm{x} \in \ker M^{n+k} \setminus \ker M^n$.
	It holds $M^{k-1}\bm{x} \in \ker M^{n+1}$ and, by assumption, $M^{k-1}\bm{x} \in \ker M^n$. 
	Then, however, $\bm{x} \in \ker M^{n+k -1} \setminus \ker M^n$, a contradiction that proves that, for any~$n$, either
	$\ker M^n \subsetneq \ker M^{n+1}$,
	or $\ker M^n = \ker M^{n+1} = \dots = \ker M^{n+k} = \dots$ holds.
	
	Recall that for every~$n$, $\ker M^n$ is a linear subspace of $\RR^d$.
	Therefore, either $\dim\ker M^{n+1} > \dim\ker M^n$, or $\dim\ker M^n = \dots = \dim \ker M^{n+k} = \dots$.
	Since $\dim\ker M^{n+i}$ is bounded from above by~$d$, the first statement follows.
	
	Finally, it suffices to check $M^n \cdot \eps \neq \bm{0}$ for $n \leq d$ in order to deduce that the equality does not hold for any~$n$.
	This gives a complete algorithm to solve the~$\bm{0}$-reachability problem.
\end{proof}
The reachability problem for $\max$- and $\min$-Bellman operators is thus completely solved for MDPs in which 
a positional scheduler with all tight actions is unique.
Equivalently, it is solved for~$\mathcal{S} = \langle M \rangle$.

However, we cannot employ the previous argument for~$\mathcal{S} = \langle M_1, M_2 \rangle$,
since the implication 
$\ker (M_1^k) = \ker(M_1^{k+1}) \Rightarrow \ker (M_2\cdot M_1^k) = \ker(M_2 \cdot M_1^{k+1})$ does not hold there.
For instance, in~\cref{ex:3d}, $\ker M_2M_1 \neq \ker M_2M_1^2$.
The decidability remains open for general semigroups---that is, when there are states with non-unique tight actions.
\section{A Decidable Case $d=2$}
\label{sec:2d}
In this section, we show that the PFR problem is decidable in dimension $d = 2$ for both $F = F_{\max}$ and $F = F_{\min}$.
Following~\cref{probequiv}, this will suffice for the decidability of the BOR problem in~$d=2$.

The key property that %
helps us establish the decidability is the existence of a total order associated with the (two-dimensional) row vectors of actions.
By arguing about this order, we show that the sequence~$\eps, F(\eps), F^2(\eps), \dots$ either has a vector comparable with~$\bm{0}$ among its first terms---or never reaches~$\bm{0}$.
\subparagraph{Actions in $d=2$ and lines.}
We further consider an MDP~$\mdp$ with~$d=2$.
Let $Act_1 = \{\alpha_1, \dots, \alpha_k\}$ and $Act_2 = \{\beta_1, \dots, \beta_\ell\}$ be the sets of tight actions available in two states~$S_d = \{s_1, s_2\}$ of~$\mdp$.
Recall that we associate the sets of row vectors $\mathcal{F}_1$ and $\mathcal{F}_2$ with actions in $Act_1$ and $Act_2$, respectively.
An action is a \emph{zero} $\alpha$- or $\beta$-action if its row vector is a zero vector~$(0, 0)$.
Denote by~$Act_i^*$ the subset of all non-zero actions in~$Act_i$, for each $i \in \{ 1, 2\}$.

We now discuss lines that correspond to actions of MDPs with $d=2$.
For our proof, it is important to identify the actions whose lines have the largest/smallest slope.
With a mild abuse of notation, we denote by $\alpha_i(x_1, x_2)$ the homogeneous part of the linear polynomial of~$\alpha_i$ (similarly for $\beta_j$).
This is exactly the dot product of action's row vector with $(x_1, x_2)$, or
\[\alpha_i(x_1, x_2) \coloneq \prob(s_1, \alpha_i, s_1) x_1 + \prob(s_1, \alpha_i, s_2) x_2.\]
For each non-zero action~$\gamma\in Act_1^* \cup Act_2^*$, the set of points~$(x_1, x_2)$ with $\gamma(x_1,x_2) = 0$ is the line orthogonal to the row vector of~$\gamma$.
We denote the \emph{angle} between this line and the positive direction of the $x$-axis measured counterclockwise by $\measuredangle\gamma$.
The angles of actions in~$Act^* = Act_1^* \cup Act_2^*$ are
numbers in $[\frac{\pi}{2}, \pi]$ and thus are totally ordered.

We denote by~$\alpha_{lo}$ an $\alpha$-action with the largest angle\footnote{It has the ``{\bf lo}west" line in the second quadrant.}.
The sign of~$\alpha_{lo}(x_1, x_2)$ is well-defined, that is,
it is independent of the choice of action
among those with the largest angle.
Similarly, $\alpha_{hi}$ is an $\alpha$-action with the smallest angle. 
Actions $\beta_{lo}$ and $\beta_{hi}$ are defined analogously.

\subparagraph{Kernels and their lines.} We now resume the matrix argumentation.
The row vectors in~$\mathcal{F}_1$ and~$\mathcal{F}_2$ define the product family $\{M_{1,1}, \dots, M_{k,\ell}\} \subset \QQ^{2 \times 2}$. That is,
\[%
M_{i,j} \cdot \bm{v} = 
\begin{pmatrix}
	\alpha_i(\bm{v}) \\	\beta_j(\bm{v})
\end{pmatrix}
\]
holds for all~$\alpha_i \in Act_1$, $\beta_j \in Act_2$, and all~$\bm{v} \in \QQ^2$.

\begin{definition}
	$F^{-1}(\bm{0})$ is the set of all vectors~$\eps$ such that $F(\eps) = \bm{0}$.
\end{definition}
We emphasise that the set~$F^{-1}(\bm{0})$ is entirely contained in the union $\cup_{i,j} \ker M_{i,j}$ of matrix kernels.
One immediate consequence thereof is that we can assume that at least one of the
matrices is singular. 
In other words, there exists~$M_{i,j}$ with
$\ker M_{i,j} \neq \{\bm{0}\}$.
Otherwise, no vector~$\bm{x} \neq \mu\Phi$ ever reaches~$\mu\Phi$.

A singular non-zero matrix~$M \in \QQ_{\geq 0}^{2\times 2}$ is of one of the two kinds.
Either one row vector of~$M$ is zero (corresponding to the case of a zero action);
or the two non-zero row vectors are collinear.
In both cases, $\ker M$ is a line with the angle in $[\frac{\pi}{2}, \pi]$.

Therefore, among all lines corresponding to the actions of~$\mdp$ 
one can identify those that are kernels of some matrix in the product family.
We refer to such lines as \emph{kernel lines}.
Let~$M_{lo}$ be a singular matrix whose kernel 
line has the largest angle among kernel lines\footnote{It is then ``below'' other kernel lines in the second quadrant.}. 
Similarly, let~$M_{hi}$ denote a singular matrix whose kernel has the smallest angle among kernel lines.
Note that the kernels of~$M_{lo}, M_{hi}$ are well-defined: there is a total order on the lines of actions, 
and there is at least one kernel line provided $F^{-1}(\bm{0}) \neq \{\bm{0}\}$.

Let $Q_i$, $i\in \{1,2,3,4\}$, be a coordinate plane quadrant.
Here, each $Q_i$ is a closed set.
For example, $Q_2 = \{\bm{x} = (x_1, x_2): x_1 \leq 0 \wedge x_2 \geq 0\}$.

\begin{lemma}\label{2d-comp-soon}
	Given a product family $\{M_{1,1}, \dots, M_{k,\ell}\} \subset \QQ^{2\times 2}$ of substochastic matrices
	and a vector~$\eps = (\varepsilon_1, \varepsilon_2)$. 
	The map~$F$ is defined as above by
	\[F(\bm{v}) = \max_{1\leq i \leq k, 1\leq j \leq \ell} \left(M_{i,j} \cdot\bm{v}\right).\]
	Exactly one of the two statements holds:
	\begin{enumerate}
		\item $F^n(\eps) \neq \bm{0}$ for all~$n$.
		\item $F^2(\eps)$ is comparable with~$\bm{0}$.
	\end{enumerate}
\end{lemma}
\begin{proof}
	We can assume~$\eps\bowtie\bm{0}$, otherwise 2.\ holds trivially.
	Without loss of generality, let $\eps = (\varepsilon_1, \varepsilon_2)$ be such that $\varepsilon_1 < 0$ and $\varepsilon_2 > 0$, i.e., $\eps \in Q_2$.
	The discussion below is driven by the question 
	\begin{center}
		``When are both vectors $F(\eps)$ and $F^2(\eps)$ incomparable with~$\bm{0}$?''.
	\end{center}
	
	Crucially, the components of~$F(\bm{v})$ can be expressed as $\max$ of actions applied:
	\begin{equation}\label{each-act-max}
		F(\bm{v}) = 
		\begin{pmatrix} \max_{1\leq i \leq k}  \alpha_i(\bm{v}) \\ \max_{1\leq j \leq \ell}  \beta_j(\bm{v})
		\end{pmatrix}
	\end{equation}
	
	We first identify one special case: when both a zero $\alpha$- and a zero $\beta$-action exist.
	Then $F(\eps)_1 \geq \bm{0}\cdot\eps = 0$ and $F(\eps)_2 \geq \bm{0}\cdot\eps = 0$,
	thus showing that $F(\eps) \geq \bm{0}$.
	We now proceed under the assumption that zero actions exist for at most one state of~$\mdp$.
	Then $F^{-1}(\bm{0})$ is contained in the union of kernel lines.
	
		Observe that for a singular non-zero matrix~$M$, the vector $M\cdot\eps$ is always comparable with~$\bm{0}$.
		Let $\theta \in [\frac{\pi}{2}, \pi]$ be the angle of the kernel line of~$M_{lo}$.
		In our case distinction, we compare the angles~$\measuredangle \alpha_{lo}$ and $\measuredangle\beta_{lo}$ with $\theta$.
		\begin{enumerate}
			\item $\theta = \max(\measuredangle \alpha_{lo}, \measuredangle \beta_{lo})$: kernel of~$M_{lo}$ is the line with the largest angle. 
			We use~\cref{each-act-max} repetitively: either we have 
			$M_{lo}\cdot\eps \geq \bm{0}$ and hence $F(\eps) \geq \bm{0}$, or $M_{lo} \cdot\eps < \bm{0}$. 
			In the latter case, $F(\eps) < \bm{0}$ 
			because $\alpha(\eps) < 0$ and $\beta(\eps) < 0$ for any non-zero $\alpha,\beta$; 
			and there is a state with no zero actions.
			Therefore, $F(\eps)$ is comparable with~$\bm{0}$.
			\item $\theta < \measuredangle\gamma$ for some action $\gamma \in Act_1^* \cup Act_2^*$: there exists an action whose line has a larger angle than all kernel lines.
			For any $M_{i,j}$ and any non-zero $\bm{x} \in \ker M_{i,j} \cap Q_2$ we have $\gamma(\bm{x}) > 0$. 
			Then either $F(\bm{x})_1 > 0$ or $F(\bm{x})_2 > 0$, and hence $F(\bm{x}) \neq \bm{0}$. 
			This yields $\ker M_{i,j} \cap Q_2 = \{\bm{0}\}$ for all~$i,j$,
			and further, $F^{-1}(\bm{0}) \cap Q_2 = \{\bm{0}\}$. 
			
			We may assume $\gamma$ is an action with the largest angle. 
			Note that its line is \emph{not} a kernel line because $\measuredangle\gamma > \theta$.
			Therefore, this line corresponds to actions of one kind (either $\alpha$-, or $\beta$-). 
			It follows that $\measuredangle\alpha_{lo} \neq \measuredangle\beta_{lo}$ and we proceed with the following case distinction.
			\begin{enumerate}
				\item $\measuredangle \beta_{lo} > \measuredangle \alpha_{lo}$. 
				We show that $F^n(\eps) \not\in Q_4$ for all~$n \geq 0$. 
				We emphasise that no zero $\alpha$-action exists, as this would imply that the line of~$\beta_{lo}$ is a kernel line.
				Now, assume towards a contradiction that~$m$ is the least integer such that $F^m(\eps) \in Q_4$.
				Consider $\bm{y}:= F^{m-1}(\eps)$.
				Since $F(\bm{y})_1 \geq 0$ and $\bm{y} \not\in Q_4$, we have $\alpha_{lo}(\bm{y}) \geq 0$.
				This, in turn, together with $\bm{y} \not\in Q_4$, implies $\beta_{lo}(\bm{y}) > 0$. Therefore, $F(\bm{y})_2 = F^m(\eps)_2 > 0$ contradicting our assumption.
				We have $F^n(\eps) \not\in Q_4$ for all~$n \geq 0$.
				
				\item $\measuredangle \alpha_{lo} > \measuredangle \beta_{lo}$.
				If $\alpha_{lo}(\eps) \leq 0$, then also $\beta_{lo}(\eps) \leq 0$ and so $F(\eps) \leq 0$.
				Moreover, if $\beta_{lo}(\eps) \geq \bm{0}$, then $F(\eps) \geq 0$.
				We now assume $\beta_{lo}(\eps) < 0 < \alpha_{lo}(\eps)$.
				Note that no zero $\beta$-action exists, for the line of $\alpha_{lo}$ is not a kernel line.  %
				Then, $F(\eps) \in Q_4 \setminus \{\bm{0}\}$.
			\end{enumerate}
		\end{enumerate}
		
		From the discussion so far, we know that either $F(\eps)$ is comparable with~$\bm{0}$, 
		or $F^n(\eps) \neq \bm{0}$ for all $n \geq 0$,
		or $F(\eps) \in Q_4\setminus \{\bm{0}\}$ while $F^{-1}(\bm{0}) \cap Q_2 = \{\bm{0}\}$.
		We now investigate the third option.
		Here, we focus our attention on actions~$\alpha_{hi}$ and $\beta_{hi}$ rather than~$\alpha_{lo}$ and $\beta_{lo}$.

		Let $\eta \in [\frac{\pi}{2}, \pi]$ be the angle of the kernel line of~$M_{hi}$.
		\begin{enumerate}
			\item $\eta = \min(\measuredangle \alpha_{hi}, \measuredangle \beta_{hi})$: kernel of~$M_{hi}$ is the line with the smallest angle. Then 
			either $M_{hi}\cdot F(\eps) \geq \bm{0}$ and hence $F^2(\eps) \geq \bm{0}$, or $M_{hi} \cdot F(\eps) < \bm{0}$. In the latter case, $F^2(\eps) < \bm{0}$ because $\alpha(F(\eps)) < 0$ and $\beta(F(\eps)) < 0$ for any non-zero $\alpha,\beta$. 
			In either case,	$F^2(\eps)$ is comparable with~$\bm{0}$.
			\item $\eta > \measuredangle\gamma$ for some action $\gamma \in Act_1^* \cup Act_2^*$: there exists an action whose line has a smaller angle than all kernel lines.
		For any $M_{i,j}$ and any $\bm{x} \in \ker M_{i,j} \cap Q_4$ we have $\gamma(\bm{x}) > 0$. Then $F(\bm{x}) \neq \bm{0}$, and hence $F^{-1}(\bm{0}) \cap Q_4 = \{\bm{0}\}$.
			
			Recall that we already proved $F^{-1}(\bm{0}) \cap Q_2 = \{\bm{0}\}$ for the case under consideration.
		Moreover, since $F^{-1}(\bm{0})$ is contained in the union of kernel lines and thus in $Q_2 \cup Q_4$, we have $F^{-1}(\bm{0}) = \{\bm{0}\}$.
		Therefore, $F^n(\eps) \neq \bm{0}$ for all~$n\geq 0$.
		\end{enumerate}
		
		In most cases we were able to show $F(\eps)$ is comparable with~$\bm{0}$. This implies $F^2(\eps)$ is comparable with~$\bm{0}$, too.
		In all other cases we either directly showed that $F^2(\eps)$ is comparable with~$\bm{0}$, or derived $F^n(\eps) \neq \bm{0}$ for all~$n$.
\end{proof}

With~\cref{2d-comp-soon}, if $\Phi^2(\bm{s}) \bowtie \mu\Phi$ for $\bm{s} \in U_\delta(\mu\Phi)$, then the answer to the BOR problem is guaranteed negative. Otherwise, we can exploit the procedure of~\cref{compar-decid}.
We derive the following result
for Bellman operators in~$d=2$ with arbitrarily many pieces.
\begin{theorem}\label{2d-efs}
	The max-BOR problem is decidable with $d=2$.
\end{theorem}
\begin{figure}[t]
	\begin{minipage}[b]{0.45\columnwidth}
		\begin{center}
			\begin{tikzpicture}[on grid,auto]
				\node[dot] (n0) {};
				\node[state] (s1) [above right = 1 and 2 of n0] {$\mathstrut s_1$};
				\node[state] (s2) [above right = -1 and 2 of n0] {$\mathstrut s_2$};
				\node[dot] (n3) [above right = -1 and 1.5 of s1] {};
				\node[] (me) [above left=1 and 1.5 of n1] {\small$\mdp_3$:};
				\node[state, accepting] (t) [above=1 of s1] {$\mathstrut t$};
				
				\path[->]
				(n0) edge [bend right] node[pos=0.6,inner sep=1pt] {1/2} (s1)
				(n0) edge [bend right] node[swap,pos=0.6,inner sep=1pt] {1/3} (s2)
				(s1) edge [bend left] node[pos=0.3,inner sep=1pt] {$\alpha_2$} (n3)
				(s2) edge [bend right] node[swap,pos=0.5,inner sep=1pt] {$\beta_1$} (n3)
				(s1) edge [bend right] node[swap,pos=0.5,inner sep=1pt] {$\alpha_1$} (n0)
				(n3) edge [bend left] node[swap,pos=0.6,inner sep=1pt] {1/2} (s1)
				(n3) edge [bend right] node[swap,pos=0.6,inner sep=1pt] {1/5} (s2)
				(n3) edge [bend right] node[swap,pos=0.6,inner sep=1pt] {1/4} (t)
				(n0) edge [bend left] node[swap,pos=0.8,inner sep=1pt] {5/36} (t)
				;
			\end{tikzpicture}
		\end{center}	
	\caption{An MDP $\mdp_3$ of~\cref{ex:2d}.}
	\label{fig:2d}
	\end{minipage}\hfill
	\begin{minipage}[b]{0.5\columnwidth}
	\begin{center}
		
		\begin{tikzpicture}
			\begin{axis}[
				axis lines=middle,
				xmin=-0.19,xmax=0.19,ymin=-0.19,ymax=0.19,
				xtick distance=0.1,
				ytick distance=0.1,
				minor tick num = 4,
				xlabel=$x$,
				ylabel=$y$,
				grid=both, 
				grid style={very thin,densely dotted,black!20}]
				\addplot [domain=0.12:-0.12,samples=2, color=red] {x*(-3)/2} node[below left]{$\alpha_1$};
				\addplot [domain=0.07:-0.07,samples=2, color=blue] {x*(-5)/2} node[right]{$\alpha_2 = \beta_1$};
				\addplot [domain=0.07:-0.07,samples=2, color=blue] {x*(-5)/2} node[right]{};
				
				\addlegendentry{$\alpha_{lo} = \alpha_1$}
				\addlegendentry{$\alpha_{hi} = \alpha_2$}		
				\addlegendentry{$\beta_{lo} = \beta_{hi} = \beta_1$}
				
				\addplot [
				only marks,
				mark=ball,
				mark size=2pt,
				point meta=explicit symbolic,
				nodes near coords,
				every node near coord/.append style={yshift=3pt,anchor=west},
				] coordinates {
					(-31/315,1/6) [$\eps$]
					(0,0)    [$\bm{0}$]
				};
			
				\addplot [
			only marks,
			mark=ball,
			mark size=2pt,
			point meta=explicit symbolic,
			nodes near coords,
			every node near coord/.append style={xshift=2pt,yshift=-3pt,anchor=west},
			] coordinates {
				(2/315, -1/63) [$F(\eps)$]
			};
			\end{axis}
		\end{tikzpicture}
	\end{center}
\caption{$\mdp_3$ illustrating Case~2(b) for~$\theta$ and Case~1 for~$\eta$.}
\label{fig:2d-plot}
\end{minipage}
\end{figure}
\begin{example}\label{ex:2d}
	We provide an example~$\mdp_3$ with $d=2$,
	where $\eps, F(\eps)$ are vectors incomparable with~$\bm{0}$,
	and $F^2(\eps) = \bm{0}$.
	Here, $\alpha_1(x_1, x_2) = \frac{1}{2}x_1 + \frac{1}{3}x_2$, $\alpha_2(x_1,x_2) = \beta_1(x_1, x_2) = \frac{1}{2}x_1 +\frac{1}{5}x_2$. 
	That is, the example corresponds to the Case~2(b) for $\theta$ and Case~1 for $\eta$:
	\[\theta = \eta = \measuredangle \beta_1 = \measuredangle \alpha_2 < \measuredangle \alpha_1.%
	\]
	Let $\eps = (-\frac{31}{315},\frac{1}{6})$. 
	Then,
	$\beta_1(\eps) = \alpha_2(\eps) < 0 < \alpha_1(\eps)$;
	 we get~$F(\eps) = (\alpha_1(\eps),\beta_1(\eps)) = (\frac{2}{315},-\frac{1}{63}) \in Q_4$. 
	Further, $\alpha_2(F(\eps)) > \alpha_1(F(\eps))$ and $F^2(\eps) = (\alpha_2(F(\eps)),\beta_1(F(\eps))) = (0,0)$.
\end{example}

\begin{lemma}\label{2d-comp-soon-min}
	Given a product family $\{M_{1,1}, \dots, M_{k,\ell}\} \subset \QQ^{2\times 2}$ of substochastic matrices
	and a vector~$\eps = (\varepsilon_1, \varepsilon_2)$. 
	Let the map~$F$ be defined by
	\[F(\bm{v}) = \min_{1\leq i \leq k, 1\leq j \leq \ell} \left(M_{i,j} \cdot\bm{v}\right).\]
	Exactly one of the two statements holds:
	\begin{enumerate}
		\item $F^n(\eps) \neq \bm{0}$ for all~$n$.
		\item $F^2(\eps)$ is comparable with~$\bm{0}$.
	\end{enumerate}
\end{lemma}
\begin{proof}
	As in the proof of~\cref{2d-comp-soon} we can assume $\eps = (\varepsilon_1, \varepsilon_2)$ such that $\varepsilon_1 < 0$ and $\varepsilon_2 > 0$, i.e., $\eps \in Q_2$.
	The components of~$F(\bm{v})$ can be expressed as $\min$ of actions applied:
	\begin{equation}\label{each-act-min}
		F(\bm{v}) = 
		\begin{pmatrix} \min_{1\leq i \leq k}  \alpha_i(\bm{v}) \\ \min_{1\leq j \leq \ell}  \beta_j(\bm{v})
		\end{pmatrix}
	\end{equation}
	
	We identify a special case: when both a zero $\alpha$- and a zero $\beta$-action exist.
	Then, for any $\eps$, we have $F(\eps)_1 \leq \bm{0}\cdot\eps = 0$ and $F(\eps)_2 \leq \bm{0}\cdot\eps = 0$, and hence $F(\eps) \leq \bm{0}$. 
	This yields the second item of the statement. 
	We now proceed under the assumption that zero actions exist for at most one state of~$\mdp$.
	Then $F^{-1}(\bm{0})$ is contained in the union of kernel lines.
	
	Let $\eta \in [\frac{\pi}{2}, \pi]$ be the angle of the kernel line of~$M_{hi}$.
	Next, we compare the angles~$\measuredangle \alpha_{hi}$ and $\measuredangle\beta_{hi}$ with $\eta$.
	\begin{enumerate}
		\item $\eta = \min(\measuredangle \alpha_{hi}, \measuredangle \beta_{hi})$: kernel of~$M_{hi}$ is the line with the smallest angle.
		 Either we have 
		 $M_{hi}\cdot\eps \leq \bm{0}$ and hence $F(\eps) \leq \bm{0}$, or $M_{hi} \cdot\eps > \bm{0}$. 
		 In the latter case, $F(\eps) > \bm{0}$ 

		 because $\alpha(\eps) > 0$ and $\beta(\eps) > 0$ for any non-zero $\alpha,\beta$; 
		 and there is a state with no zero actions.
		 Therefore, $F(\eps)$ is comparable with~$\bm{0}$.
		\item $\eta > \measuredangle \gamma$ for some action $\gamma \in Act^*_1 \cup Act^*_2$. 
		For any $M_{i,j}$ and any non-zero $\bm{x} \in \ker M_{i,j} \cap Q_2$ we have $\gamma(\bm{x}) < 0$. 
		Then either $F(\bm{x})_1 < 0$ or $F(\bm{x})_2 < 0$, and hence $F(\bm{x}) \neq \bm{0}$. 
		This yields $\ker M_{i,j} \cap Q_2 = \{\bm{0}\}$ for all~$i,j$,
		and further, $F^{-1}(\bm{0}) \cap Q_2 = \{\bm{0}\}$. 
		
		Similarly to the proof of~\cref{2d-comp-soon}, 
		only actions of one kind have angle~$\measuredangle\gamma$ 
		if $\measuredangle\gamma$ is chosen to be minimal;
		hence we derive $\measuredangle\alpha_{hi} \neq \measuredangle\beta_{hi}$. 
		We now identify two complementary subcases:
		\begin{enumerate}
			\item $\measuredangle\beta_{hi} > \measuredangle\alpha_{hi}$. We show that $F^n(\eps) \not\in Q_4$ for all~$n \geq 0$.
			We emphasise that no zero $\beta$-action exists, as this would imply that the line of $\alpha_{hi}$ is a kernel line.
			Now, assume towards a contradiction that~$m$ is the least integer such that $F^m(\eps) \in Q_4$.
			Consider $\bm{y}:= F^{m-1}(\eps)$.
			Since $F(\bm{y})_2 \leq 0$ and $\bm{y} \not\in Q_4$, we have $\beta_{hi}(\bm{y}) \leq 0$.
			This, in turn, together with $\bm{y} \not\in Q_4$, implies $\alpha_{hi}(\bm{y}) < 0$. Therefore, $F(\bm{y})_1 < 0$ contradicting our assumption.
			We have $F^n(\eps) \not\in Q_4$ for all~$n \geq 0$.
			
			\item $\measuredangle\alpha_{hi} > \measuredangle\beta_{hi}$. If $\beta_{hi}(\eps) \geq 0$, then $\alpha(\eps) \geq 0$ and $\beta(\eps) \geq 0$ for all~$\alpha, \beta$ and so $F(\eps) \geq 0$.
			Moreover, if $\alpha_{hi}(\eps) \leq 0$, then $\beta_{hi}(\eps) \leq 0$ and $F(\eps) \leq \bm{0}$.
			We now assume $\beta_{hi}(\eps) < 0 < \alpha_{hi}(\eps)$. 
			No zero $\alpha$-action exists because the line of $\beta_{hi}$ is not a kernel line. 
			Therefore, $F(\eps) \in Q_4\setminus\{\bm{0}\}$.
		\end{enumerate}
	\end{enumerate}

	So far, either $F(\eps)$ is comparable with~$\bm{0}$, 
	or $F^n(\eps) \neq \bm{0}$ for all $n \geq 0$,
	or $F(\eps) \in Q_4\setminus \{\bm{0}\}$ while $F^{-1}(\bm{0}) \cap Q_2 = \{\bm{0}\}$.
	We now investigate the third option.
	Here, we focus our attention on actions~$\alpha_{lo}$ and $\beta_{lo}$.
	
	Let $\theta \in [\frac{\pi}{2}, \pi]$ be the angle of the kernel line of~$M_{lo}$.
	\begin{enumerate}
		\item $\theta = \max(\measuredangle \alpha_{lo}, \measuredangle \beta_{lo})$: kernel of~$M_{lo}$ is the line with the largest angle. Then 
		either $M_{lo}\cdot F(\eps) \leq \bm{0}$ and hence $F^2(\eps) \leq \bm{0}$, or $M_{lo} \cdot F(\eps) > \bm{0}$. In the latter case, $F^2(\eps) > \bm{0}$ because $\alpha(F(\eps)) > 0$ and $\beta(F(\eps)) > 0$ for any non-zero $\alpha,\beta$; 
		and there is a state with no zero actions.
		\item $\theta < \measuredangle\gamma$ for some action $\gamma \in Act_1^* \cup Act_2^*$: there exists an action whose line has a larger angle than all kernel lines.
		For any $M_{i,j}$ and any $\bm{x} \in \ker M_{i,j} \cap Q_4$ we have $\gamma(\bm{x}) < 0$. Then $F(\bm{x}) \neq \bm{0}$, and hence $F^{-1}(\bm{0}) \cap Q_4 = \{\bm{0}\}$. 
		
		We proved earlier $F^{-1}(\bm{0}) \cap Q_2 = \{\bm{0}\}$ for the case under consideration.
		Since $F^{-1}(\bm{0})$ is contained in the union of kernel lines and thus in $Q_2 \cup Q_4$, we have $F^{-1}(\bm{0}) = \{\bm{0}\}$.
		Therefore, $F^n(\eps) \neq \bm{0}$ for all~$n\geq 0$.
	\end{enumerate}
	
	In most cases we were able to show $F(\eps)$ is comparable with~$\bm{0}$. This implies $F^2(\eps)$ is comparable with~$\bm{0}$, too.
	In all other cases we either directly showed that $F^2(\eps)$ is comparable with~$\bm{0}$, or derived $F^n(\eps) \neq \bm{0}$ for all~$n$.
\end{proof}
Similarly to the $\max$-case, 
we may use~\cref{2d-comp-soon-min} to show full decidability of $\min$-BOR in the two-dimensional case. 
If $\Phi^2(\bm{s}) \bowtie \mu\Phi$ for $\bm{s} \in U_\delta(\mu\Phi)$, then the answer to the BOR problem is guaranteed negative. Otherwise, we can exploit the procedure of~\cref{compar-decid}.
\begin{theorem}\label{2d-bor-min}
	The min-BOR problem is decidable with $d=2$.
\end{theorem}

\section{Related Work and Discussion}
\label{relatedwork}
In our work, we have outlined a series of phenomena inherent to the iterative application of Bellman operators. 
Crucial was the fact that the update coefficients were non-negative.

Notice that any PAM on the domain~$[0,1]^d$ can be represented as a 
nested min-max of its affine components~\cite{Ovchinnikov2002maxmin,Gorokhovik1994maxmin}.
This supports the relevance of PAMs with updates defined by the maximum/minimum of affine pieces, 
Bellman operators of MDPs being among them.
Generalising our results beyond Bellman operators of MDPs 
using monotonicity and fixed-point convergence is a subject of our future work.
In fact, there are characterisations of Bellman operators by convex polytopes. 
\begin{proposition}[characterisation of $\Phi_{\max}$]
	Given a function $F\colon [0, 1]^d \rightarrow [0, 1]^d$, 
	the following are equivalent:
	\begin{enumerate}
		\item $F$ is a Bellman operator $\Phi_{\max}$ of an MDP.
		\item for any $i\in \{1,\dots, d\}$, its epigraph $\mathrm{epi}(\pi_i \circ F)\subseteq [0, 1]^{d+1}$ is a convex polytope, and 
		it can be written by 
		\[
		\mathrm{epi}(\pi_i \circ F) = \{(\bm{x}, t)\in [0, 1]^d\times [0, 1] \mid (\bm{a}_i)^{\top}\cdot \bm{x} + b_i \leq t,\, \forall i\in I\}
		\]
		with a finite set $I$, a vector $\bm{a}_i\in [0, 1]^d$, and a scalar $b_i\in [0, 1]$, where $b_i + \sum_{j} (a_{i})_j \leq 1$. 
	\end{enumerate}
\end{proposition}

A similar characterisation for $\Phi_{\min}$ is also possible. For instance, removing the assumption of $b_i + \sum_{j} (a_{i})_j \leq 1$ and seeing whether 
we can still get the same decidability results would be of interest. 

It is worth pointing out that we capture a large class of PAMs 
for which the assumptions of the known techniques do not hold.
In proving decidability for~$d=2$,
we do not impose the restrictions used in works on the one-dimensional version of the problem.
Bellman operators are, in general, neither injective as in~\cite{ghahremani2023injective} nor complete~\cite{bournez2018onedpam}, nor even surjective.

Our techniques and results can be applied to other problems about PAMs.
Consider, in particular, the universally quantified version of~BOR:
given a vector $\bm{t} \in [0,1]^d\cap \QQ$ and a Bellman operator~$\Phi: [0,1]^d \rightarrow [0,1]^d$,
does there exist~$n \in \NN$ \emph{for every}~$\bm{s} \in [0,1]^d\cap \QQ$ such that $\Phi^n(\bm{s}) = \bm{t}$?
This is the mortality problem, known to be undecidable for general PAMs in dimension~2~\cite{BenAmram2015mortal,blondel2001note}.
For Bellman operators, however, it is equivalent to solving the BOR problem for $\bm{s} = \bm{0}$ and for $\bm{s} = \bm{1}$.
Indeed, both BOR instances are ``yes''-instances if and only if every instance with $\bm{0} \leq \bm{s} \leq \bm{1}$ is a~``yes''.
We can answer this using our novel algorithm (\cref{compar-decid}), hence mortality for Bellman operators is decidable in all dimensions.

Without restrictions on the dimension~$d$,
the decidability of the BOR problem remains open.
Notably, if there exist states with multiple tight actions,
then the dynamics of the sequence $\langle F^n(\eps)\rangle_{n\in \NN}$ as defined in our PFR discussion is intricate for~$\eps \bowtie \bm{0}$.
For more works discussing the iterative dynamics of map~$F$ see~\cite{Sladky1980max,Zijm1984max}.
We highlight that
$F^n(\eps)$ does not have a closed form, in contrast to $M^n\eps$ for a fixed matrix~$M$.
While the entries of the vector~$M^n\eps$ are terms of linear recurrence sequences, 
the class of such sequences is not closed under max~\cite{Heerdt2018maxlrs}.
The behaviour of entries in~$F^n(\eps)$ is even subtler than that, as 
witnessed by the non-positionality discussion and~\cref{ex:3d}.

We emphasise that due to the undecidability of the vector reachability problem
for matrix semigroups~\cite{Bell2008undec}, following~\cref{not-just-reach} we need to argue about ultimate non-positivity of~$F^n(\eps)$. %
However, even for one~$M$, deciding whether $(M^n\eps)_1 \leq 0$ for some~$n$ is equivalent to the positivity problem~\cite{Ouaknine2014pos5}, open in dimension $d > 5$.
This hardness carries on to the stochastic matrices~\cite{Mihir2024ergodic}.
Moreover, the question whether there exists~$n$ 
such that $M^n\eps \leq \bm{0}$
corresponds to the polyhedron-hitting problem, and is Diophantine-hard~\cite{ChonevOW15} for general matrices~$M$.

\paragraph{Acknowledgement}
We would like to express our gratitude to the anonymous reviewers for their valuable comments and suggestions.
We further thank Ichiro Hasuo and Laura Kovács for their helpful feedback.
We also thank James Worrell for his suggestion to investigate the characterisations of Bellman operators by the convexity of their graphs.

\includegraphics[height=1em]{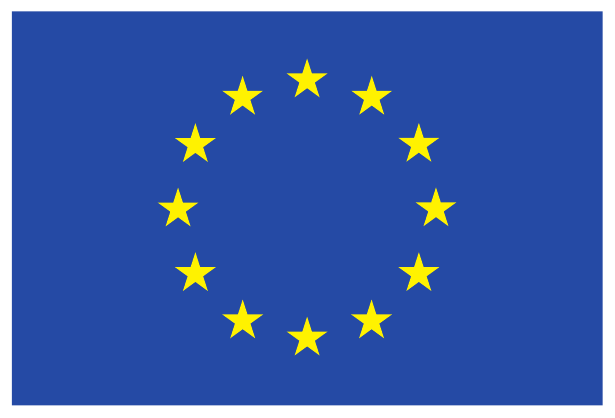} This paper is part of a project that has received funding from the European Research Council (ERC) under the European Union's Horizon 2020 research and innovation program (grant agreement No.~10103444).
The authors acknowledge the support of the ERC consolidator grant ARTIST 101002685. 
K.\ W.\ is supported by the JST grants No. JPMJAX23CU and JPMJPR25KD.

\bibliographystyle{plain}
\bibliography{references}
\end{document}